\newif\ifdoublecol

\doublecoltrue
\documentclass[twocolumn]{IEEEtran}
\IEEEoverridecommandlockouts

\usepackage{etex} 

\ifCLASSINFOpdf
  \usepackage{graphicx}
  
  \usepackage{epstopdf}
\else
    \usepackage{graphicx}

  \graphicspath{{../eps/}}
  \DeclareGraphicsExtensions{.eps}
\fi
\usepackage{cite}

\usepackage[cmex10]{amsmath}
\usepackage{xcolor,colortbl}
\definecolor{col1}{HTML}{3891A6}
\definecolor{col2}{HTML}{EF5B5B}
\definecolor{col3}{HTML}{3DDC97}

\usepackage{amsmath,amsthm,relsize}
\usepackage{amsfonts, amssymb, cuted}
\usepackage{mathtools}
\usepackage{algorithm}
\usepackage[noend]{algpseudocode}
\usepackage{cite}
\usepackage{soul}
\usepackage{tikz}
\usepackage{pgfplotstable}
\usepackage{pgfplots}
\usepackage{nicefrac}
\usepackage{enumitem}
\usepackage{support-caption}
\usepackage{subcaption}
\usepackage[font=footnotesize]{caption}
\usepackage{multirow}
\pgfplotsset{compat=1.15}
\usepackage{relsize}
\usetikzlibrary{patterns}
\usepackage{acro}
\usepackage{pifont}
\newtheorem{theorem}{Theorem}

\usepackage{todonotes}

\usepackage{mathtools}

\usepackage{array}
\usepackage{booktabs}

\usepackage{diagbox}
\usepackage{smartdiagram}
\usesmartdiagramlibrary{additions}
\usepackage{bm}
\usepackage{multicol}

\usepackage{tabularx}
\usepackage{colortbl}

\usepackage{hyperref}
\hypersetup{
    colorlinks=true,
    linkcolor=blue,
    citecolor=blue,
    filecolor=blue,
    urlcolor=blue
}

\usetikzlibrary{plotmarks}
  \usetikzlibrary{arrows.meta}
  \usepgfplotslibrary{patchplots}
  \usepackage{grffile}
  \pgfplotsset{plot coordinates/math parser=false}
  \newlength\figureheight
  \newlength\figurewidth
   \pgfplotsset{compat=1.11,
    /pgfplots/ybar legend/.style={
    /pgfplots/legend image code/.code={%
       \draw[##1,/tikz/.cd,yshift=-0.25em]
        (0cm,0cm) rectangle (3em,8pt);},
   },
}

\usepgfplotslibrary{groupplots,units}
\pgfplotsset{
  compat=1.9,
  unit code/.code 2 args={\si{#1#2}} 
}
\usepackage{siunitx}

\usepackage{color}
\usepackage{colortbl}
\usepackage{multirow}

\newlength{\Oldarrayrulewidth}

\definecolor{intnull}{RGB}{213,229,255}
\definecolor{inteins}{RGB}{128,179,255}
\definecolor{intzwei}{RGB}{42,127,255}
\definecolor{intdrei}{RGB}{0,85,212}
\definecolor{intvier}{RGB}{0,51,128}
\definecolor{intfunf}{RGB}{0,34,85}
\usepackage{arydshln} 

\usetikzlibrary{decorations.pathreplacing}
\renewcommand*{\arraystretch}{.4}

\newtheorem{lemma}{Lemma}

\newtheorem{remarknum}{Remark} 
\usepackage{arydshln}

\newcommand{\vbar}{\raisebox{.17ex}{\rule{.04em}{1.35ex}}}
\newcommand{\vbarind}{\raisebox{.01ex}{\rule{.04em}{1.1ex}}}
\newcommand{\R}{\ifmmode{\rm I}\hspace{-.2em}{\rm R} \else ${\rm I}\hspace{-.2em}{\rm R}$ \fi}
\newcommand{\T}{\ifmmode{\rm I}\hspace{-.2em}{\rm T} \else ${\rm I}\hspace{-.2em}{\rm T}$ \fi}
\newcommand{\N}{\ifmmode{\rm I}\hspace{-.2em}{\rm N} \else \mbox{${\rm I}\hspace{-.2em}{\rm N}$} \fi}
\newcommand{\B}{\ifmmode{\rm I}\hspace{-.2em}{\rm B} \else \mbox{${\rm I}\hspace{-.2em}{\rm B}$} \fi}
\newcommand{\Hil}{\ifmmode{\rm I}\hspace{-.2em}{\rm H} \else \mbox{${\rm I}\hspace{-.2em}{\rm H}$} \fi}
\newcommand{\C}{\ifmmode\hspace{.2em}\vbar\hspace{-.31em}{\rm C} \else \mbox{$\hspace{.2em}\vbar\hspace{-.31em}{\rm C}$} \fi}
\newcommand{\Cind}{\ifmmode\hspace{.2em}\vbarind\hspace{-.25em}{\rm C} \else \mbox{$\hspace{.2em}\vbarind\hspace{-.25em}{\rm C}$} \fi}
\newcommand{\Q}{\ifmmode\hspace{.2em}\vbar\hspace{-.31em}{\rm Q} \else \mbox{$\hspace{.2em}\vbar\hspace{-.31em}{\rm Q}$} \fi}
\newcommand{\Z}{\ifmmode{\rm Z}\hspace{-.28em}{\rm Z} \else ${\rm Z}\hspace{-.28em}{\rm Z}$ \fi}

\newtheorem{lem}{Lemma}

\newtheorem{exmp}{Example}
\theoremstyle{definition}

\newcommand{\CF}[0]{{\mathcal{F}}}

\newcommand{\CI}[0]{{\mathcal{I}}}
\newcommand{\CJ}[0]{{\mathcal{J}}}
\newcommand{\CK}[0]{{\mathcal{K}}}

\newcommand{\CM}[0]{{\mathcal{M}}}

\newcommand{\CP}[0]{{\mathcal{P}}}

\newcommand{\CT}[0]{{\mathcal{T}}}

\newcommand{\Bw}[0]{{\mathbf{w}}}
\newcommand{\Bx}[0]{{\mathbf{x}}}
\newcommand{\By}[0]{{\mathbf{y}}}
\newcommand{\Bz}[0]{{\mathbf{z}}}

\newcommand{\BH}[0]{{\mathbf{H}}}

\newcommand{\BJ}[0]{{\mathbf{J}}}

\newcommand{\BU}[0]{{\mathbf{U}}}

\newcommand{\SfG}[0]{{\mathsf{G}}}

\usepackage{tikz}

\usepackage{tikz}
\usetikzlibrary{arrows,
                calc,
                decorations.pathreplacing,
                calligraphy,
                matrix,
                positioning
                }

\DeclareAcronym{ADMM}{
    short = ADMM,
    long = alternating direction method of multipliers,
    list = Alternating Direction Method of Multipliers,
    tag = abbrev
}

\DeclareAcronym{AoA}{
    short = AoA,
    long = angle-of-arrival,
    list = Angle-of-Arrival,
    tag = abbrev
}

\DeclareAcronym{SISO}{
    short = SISO,
    long = single-input single-output,
    list = single-input single-output,
    tag = abbrev
}

\DeclareAcronym{MRT}{
    short = MRT,
    long = maximum ratio transmitter,
    list = maximum ratio transmitter,
    tag = abbrev
}

\DeclareAcronym{PDA}{
    short = PDA,
    long = placement delivery array,
    list = placement delivery array,
    tag = abbrev
}

\DeclareAcronym{EE}{
    short = EE,
    long = energy efficiency,
    list = energy efficiency,
    tag = abbrev
}

\DeclareAcronym{MDS}{
    short = MDS,
    long = maximum distance separation,
    list = maximum distance separation,
    tag = abbrev
}

\DeclareAcronym{SIC}{
    short = SIC,
    long = successive-interference-cancellation,
    list = successive-interference-cancellation,
    tag = abbrev
}

\DeclareAcronym{MAC}{
    short = MAC,
    long = multiple-access-channel,
    list = multiple-access-channel,
    tag = abbrev
}

\DeclareAcronym{AoD}{
    short = AoD,
    long = angle-of-departure,
    list = Angle-of-Departure,
    tag = abbrev
}

\DeclareAcronym{BB}{
    short = BB,
    long = base band,
    list = Base Band,
    tag = abbrev
}

\DeclareAcronym{BC}{
    short = BC,
    long = broadcast channel,
    list = Broadcast Channel,
    tag = abbrev
}

\DeclareAcronym{BS}{
    short = BS,
    long = base station,
    list = Base Station,
    tag = abbrev
}

\DeclareAcronym{BR}{
    short = BR,
    long = best response,
    list = Best Response, 
    tag = abbrev
}

\DeclareAcronym{CB}{
    short = CB,
    long = coordinated beamforming,
    list = Coordinated Beamforming,
    tag = abbrev
}

\DeclareAcronym{CC}{
    short = CC,
    long = coded caching,
    list = Coded Caching,
    tag = abbrev
}

\DeclareAcronym{CE}{
    short = CE,
    long = channel estimation,
    list = Channel Estimation,
    tag = abbrev
}

\DeclareAcronym{CoMP}{
    short = CoMP,
    long = coordinated multi-point transmission,
    list = Coordinated Multi-Point Transmission,
    tag = abbrev
}

\DeclareAcronym{CRAN}{
    short = C-RAN,
    long = cloud radio access network,
    list = Cloud Radio Access Network,
    tag = abbrev
}

\DeclareAcronym{CSE}{
    short = CSE,
    long = channel specific estimation,
    list = Channel Specific Estimation,
    tag = abbrev
}

\DeclareAcronym{CSI}{
    short = CSI,
    long = channel state information,
    list = Channel State Information,
    tag = abbrev
}

\DeclareAcronym{CSIT}{
    short = CSIT,
    long = channel state information at the transmitter,
    list = Channel State Information at the Transmitter,
    tag = abbrev
}

\DeclareAcronym{CU}{
    short = CU,
    long = central unit,
    list = Central Unit,
    tag = abbrev
}

\DeclareAcronym{D2D}{
    short = D2D,
    long = device-to-device,
    list = Device-to-Device,
    tag = abbrev
}

\DeclareAcronym{DE-ADMM}{
    short = DE-ADMM,
    long = direct estimation with alternating direction method of multipliers,
    list = Direct Estimation with Alternating Direction Method of Multipliers,
    tag = abbrev
}

\DeclareAcronym{DE-BR}{
    short = DE-BR,
    long = direct estimation with best response,
    list = Direct Estimation with Best Response,
    tag = abbrev
}

\DeclareAcronym{DE-SG}{
    short = DE-SG,
    long = direct estimation with stochastic gradient,
    list = Direct Estimation with Stochastic Gradient,
    tag = abbrev
}

\DeclareAcronym{DFT}{
	short = DFT,
	long = discrete fourier transform,
	list = Discrete Fourier Transform,
	tag = abbrev
}

\DeclareAcronym{DoF}{
    short = DoF,
    long = degrees of freedom,
    list = Degrees of Freedom,
    tag = abbrev
}

\DeclareAcronym{DL}{
    short = DL,
    long = downlink,
    list = Downlink,
    tag = abbrev
}

\DeclareAcronym{GD}{
	short = GD, 
	long = gradient descent,
	list = Gradeitn Descent,
	tag = abbrev
}

\DeclareAcronym{IBC}{
    short = IBC,
    long = interfering broadcast channel,
    list = Interfering Broadcast Channel,
    tag = abbrev
}

\DeclareAcronym{i.i.d.}{
    short = i.i.d.,
    long = independent and identically distributed,
    list = Independent and Identically Distributed,
    tag = abbrev
}

\DeclareAcronym{JP}{
    short = JP,
    long = joint processing,
    list = Joint Processing,
    tag = abbrev
}

\DeclareAcronym{KKT}{
    short = KKT,
    long = Karush-Kuhn-Tucker,
    tag = abbrev
}

\DeclareAcronym{LOS}{
	short = LOS,
	long = line-of-sight,
	list = Line-of-Sight,
	tag = abbrev
}

\DeclareAcronym{LS}{
    short = LS,
    long = least squares,
    list = Least Squares,
    tag = abbrev
}

\DeclareAcronym{LTE}{
    short = LTE,
    long = Long Term Evolution,
    tag = abbrev
}

\DeclareAcronym{LTE-A}{
    short = LTE-A,
    long = Long Term Evolution Advanced,
    tag = abbrev
}

\DeclareAcronym{MIMO}{
    short = MIMO,
    long = multiple-input multiple-output,
    list = Multiple-Input Multiple-Output,
    tag = abbrev
}

\DeclareAcronym{MISO}{
    short = MISO,
    long = multiple-input single-output,
    list = Multiple-Input Single-Output,
    tag = abbrev
}

\DeclareAcronym{MSE}{
    short = MSE,
    long = mean-squared error,
    list = Mean-Squared Error,
    tag = abbrev
}

\DeclareAcronym{MMSE}{
    short = MMSE,
    long = minimum mean-squared error,
    list = Minimum Mean-Squared Error,
    tag = abbrev
}

\DeclareAcronym{mmWave}{
	short = mmWave,
	long = millimeter wave,
	list = Millimeter Wave,
	tag = abbrev
}

\DeclareAcronym{MU-MIMO}{
    short = MU-MIMO,
    long = multi-user \ac{MIMO},
    list = Multi-User \ac{MIMO},
    tag = abbrev
}

\DeclareAcronym{OTA}{
    short = OTA,
    long = over-the-air,
    list = Over-the-Air,
    tag = abbrev
}

\DeclareAcronym{PSD}{
    short = PSD,
    long = positive semidefinite,
    list = Positive Semidefinite,
    tag = abbrev
}

\DeclareAcronym{QoS}{
	short = QoS,
	long = quality of service,
	list = Quality of Service,
	tag = abbrev
}

\DeclareAcronym{RCP}{
	short = RCP,
	long = remote central processor,
	list = Remote Central Processor,
	tag = abbrev
}

\DeclareAcronym{RRH}{
    short = RRH,
    long = remote radio head,
    list = Remote Radio Head,
    tag = abbrev
}

\DeclareAcronym{RSSI}{
    short = RSSI,
    long = received signal strength indicator,
    list = Received Signal Strength Indicator,
    tag = abbrev
}

\DeclareAcronym{RX}{
	short = RX,
	long = receiver,
	list = Receiver,
	tag = abbrev
}

\DeclareAcronym{SCA}{
    short = SCA,
    long = successive convex approximation,
    list = Successive Convex Approximation,
    tag = abbrev
}

\DeclareAcronym{SG}{
    short = SG,
    long = stochastic gradient,
    list = Stochastic Gradient,
    tag = abbrev
}

\DeclareAcronym{SNR}{
    short = SNR,
    long = signal-to-noise ratio,
    list = Signal-to-Noise Ratio,
    tag = abbrev
}

\DeclareAcronym{SINR}{
    short = SINR,
    long = signal-to-interference-plus-noise ratio,
    list = Signal-to-Interference-plus-Noise Ratio,
    tag = abbrev
}

\DeclareAcronym{SOCP}{
	short = SOCP, 
	long = second order cone program,
	list = Second Order Cone Program,
	tag = abbrev
}

\DeclareAcronym{SSE}{
    short = SSE,
    long = stream specific estimation,
    list = Stream Specific Estimation,
    tag = abbrev
}

\DeclareAcronym{SVD}{
	short = SVD,
	long = singular value decomposition,
	list = Singular Value Decomposition,
	tag = abbrev
}

\DeclareAcronym{TDD}{
	short = TDD,
	long = time division duplex,
	list = Time Division Duplex,
	tag = abbrev
}

\DeclareAcronym{TX}{
	short = TX,
	long = transmitter,
	list = Transmitter,
	tag = abbrev
}

\DeclareAcronym{UE}{
    short = UE,
    long = user equipment,
    list = User Equipment,
    tag = abbrev
}

\DeclareAcronym{UL}{
    short = UL,
    long = uplink,
    list = Uplink,
    tag = abbrev
}

\DeclareAcronym{ULA}{
	short = ULA,
	long = uniform linear array,
	list = Uniform Linear Array,
	tag = abbrev
}

\DeclareAcronym{UPA}{
    short = UPA,
    long = uniform planar array,
    list = Uniform Planar Array,
    tag = abbrev
}

\DeclareAcronym{WMMSE}{
    short = WMMSE,
    long = weighted minimum mean-squared error,
    list = Weighted Minimum Mean-Squared Error,
    tag = abbrev
}

\DeclareAcronym{WMSEMin}{
    short = WMSEMin,
    long = weighted sum \ac{MSE} minimization,
    list = Weighted sum \ac{MSE} Minimization,
    tag = abbrev
}

\DeclareAcronym{WBAN}{
	short = WBAN,
	long = wireless body area network,
	list = Wireless Body Area Network,
	tag = abbrev
}

\DeclareAcronym{WSRMax}{
    short = WSRMax,
    long = weighted sum rate maximization,
    list = Weighted Sum Rate Maximization,
    tag = abbrev
}

\usepackage{makecell}

\usepackage{dblfloatfix}
\usepackage{placeins} 

\begin{document}

\bstctlcite{BSTcontrol}


\title{\color{black}Cache-Aided Asymmetric MIMO Communications: Achievable DoF Analysis}\color{black}


\author{\IEEEauthorblockN{Mohammad NaseriTehrani,~\IEEEmembership{Member,~IEEE,}  
MohammadJavad Salehi,~\IEEEmembership{Member,~IEEE,} 
~and Antti T\"olli}~\IEEEmembership{Senior~Member,~IEEE}
\thanks{
The authors are affiliated with the University of Oulu, Finland (emails: 
\{firstname.lastname@oulu.fi\}).
This work was supported by Infotech Oulu and by the Research Council of Finland under grants no. 343586 (CAMAIDE) and 346208 (6G Flagship).
%
This article has been presented in part in~\cite{naseritehrani2025achievable}.
}
}



\maketitle





\begin{abstract}
Integrating coded caching (CC) into multiple-input multiple-output (MIMO) communications significantly enhances the achievable degrees of freedom (DoF). This paper investigates a practical cache-aided asymmetric MIMO configuration with cache ratio $\gamma$, where a server with $L$ transmit antennas communicates with $K$ users. The users are partitioned into $J$ groups, and each user in group $j$ has $G_j$ receive antennas. We propose four content-aware MIMO-CC strategies: \emph{min-G} enforces symmetry using the smallest antenna count among users; \emph{Grouping} maximizes intra-subset spatial multiplexing gain at the expense of some global caching gain; \emph{Super-grouping} aggregates users into optimized \emph{min-$G$}-based super-sets with identical effective receive multiplexing gains before applying \emph{Grouping} across them; and \emph{Phantom} redistributes spatial resources assuming ``phantom'' antennas at the users to bridge the performance gains of \emph{min-$G$} and \emph{Grouping}.
We develop these asymmetric strategies under three reference symmetric CC placement–delivery policies with guaranteed linear decodability: a \color{black} finite-search DoF-optimized policy 
attaining the best single-shot
achievable DoF lower bound, \color{black}
and two closed-form policies, namely combinatorial and  linear cyclic low-complexity constructions, with the cyclic policy attaining DoF performance close to the others in many operating regimes.
Analytical and numerical results demonstrate significant DoF improvements across various system configurations, and that policy–strategy combinations offer flexible trade-offs between DoF and subpacketization complexity. 
\end{abstract}


\section{Introduction}


%
The demand for multimedia content, including immersive viewing, gaming, and extended reality (XR), continues to increase mobile data traffic.
With stringent requirements for ultra‑low latency ($<$10ms) and multi‑Gbps throughput levels, conventional scaling in spectrum or densification alone is insufficient to meet these requirements.
Instead, proactive caching of popular content at end-users, by offloading traffic during peak hours, is foreseen to play a critical role in reducing bandwidth pressure and network congestion, thereby enabling scalable network performance~\cite{salehi2022enhancing}.
In this regard, coded caching (CC)~\cite{maddah2014fundamental} has emerged as an effective solution, utilizing the onboard memory of network devices as a communication resource, especially beneficial for cacheable multimedia content.
The performance gain in CC arises from multicasting well-constructed codewords to user groups of size $t+1$, where the CC gain $t$ is proportional to the cumulative cache size of all users. 
Originally designed for single-input single-output (SISO) setups~\cite{maddah2014fundamental}, CC was later shown to be effective in multiple-input single-output (MISO) systems, demonstrating that spatial multiplexing and CC gains are additive~\cite{shariatpanahi2018physical}.
This is achieved by serving 
multiple groups of users simultaneously with multiple multicast messages and suppressing the intra-group interference by beamforming.
Accordingly, in a MISO-CC setting with $L$ Tx antennas, $t+ L$ users can be served in parallel, and the so-called degree-of-freedom (DoF) of $t+L$ is achievable~\cite{shariatpanahi2016multi,shariatpanahi2018physical,lampiris2021resolving,salehi2020lowcomplexity}. 


The global caching gain scales linearly with the number of users in SISO and MISO settings. However, it was only recently demonstrated that multi-antenna reception at the user side can \emph{multiplicatively boost} this gain by a factor of the spatial multiplexing dimension $G$~\cite{salehi2021MIMO,naseritehrani2024multicast,naseritehrani2026cache}. This multiplicative scaling broadens the design space for  collaborative multi-user multi-input multi-output (MIMO) deployments by enabling parallel coded multicast transmission, jointly exploiting spatial multiplexing and cache-induced side information.

All prior art, however, assumes the same number of antennas at each receiver, which may not be the case in practice. For example, 5G NR-capable devices, ranging from high-performance smartphones to low-power IoT nodes, are designed to meet diverse requirements, in terms of latency, data rates, and QoE, resulting in category-specific asymmetric antenna configurations and capabilities~\cite{dahlman20205g,salehi2022enhancing}.
Despite its practical significance, the achievable single-shot DoF in scenarios where users are equipped with asymmetric antennas has not yet been characterized in the MIMO-CC literature. To address this gap, we develop a set of cache-aided transmission schemes that explicitly account for heterogeneous user antenna configurations. 
These schemes are designed to exploit the unique characteristics of heterogeneous user equipment by carefully balancing two key performance dimensions: the spatial multiplexing gain enabled by multi-antenna receivers, and the global caching gain arising from content reuse. 

Two primary and two hybrid strategies for asymmetric MIMO-CC communications are studied in this paper. In the context of primary strategies, the proposed \emph{min-$G$} strategy emphasizes the aggregate CC gain across users by adapting delivery based on the user with the lowest spatial rank, and the \emph{Grouping} strategy employs a rank-aware approach that clusters users with similar spatial capabilities (spatial rank or link qualities), thereby maximizing the impact of receiver-side spatial multiplexing. The primary strategies are followed by \emph{Super-grouping} and \emph{Phantom} hybrid strategies, that strike a tradeoff between spatial and global caching gains. 
Collectively, these strategies broaden the design space of MIMO-CC systems by moving beyond the conventional symmetry assumption, enabling scalable content delivery in heterogeneous user environments.
\subsection{Related Work}
\label{rel_work_subs}




\subsubsection{Asymmetric SISO-CC Networks}
\label{rel_work_subs_1}
\color{black}
Early CC works mainly considered symmetric, or homogeneous, settings, where all users have equal cache capacities and share a uniform delivery link~\cite{maddah2014fundamental}. Later studies addressed several asymmetric SISO-CC settings, including cache-size heterogeneity~\cite{Karamchandani2016Hierarchical,li2016rate,amiri2017Decentralized,Yang2019TwoUser}, heterogeneous link qualities~\cite{zhang2017wireless}, joint cache--link heterogeneity~\cite{cao2019coded,bidokhti2018noisy,Amiri2018Erasure}, and more general heterogeneous SISO settings with unequal cache sizes, file lengths, and user-specific demands~\cite{Yang2018TIT,chang2022coded}.
These works motivate asymmetric CC, but their asymmetries arise from storage, link, demand, or file-size heterogeneity.
\color{black}
\subsubsection{Asymmetric MISO-CC Networks}
\label{rel_work_subs_2}

Recent advances in MISO-CC communications have begun to explore multi-antenna coded caching for networks with different cache sizes at the users. Specifically,~\cite{lampiris2020full} demonstrated that multiple transmit antennas, while providing full multiplexing gains, can simultaneously eliminate the performance penalties typically associated with cache-size imbalance, and~\cite{Urabe2022SPAWC} studied multi-user MISO systems with cache-size heterogeneity in the finite-SNR region. 
More recently,~\cite{mahmoodi2023multi} developed a connectivity-aware cache allocation that assigns larger caches to weaker users to enhance local gain, while preserving cache overlap across users to sustain global coded multicasting opportunities. 
Later, dynamic CC schemes~\cite{abolpour2024resource}
addressed asymmetric user-to-caching-profile associations and mitigated DoF loss through finite-SNR transmission design.
However, all these works still assume single-antenna receivers, and thus do not exploit
receive-side spatial dimensions.

\subsubsection{\color{black}MIMO-CC\color{black}}
\label{rel_work_subs_3}
Applying 
CC in MIMO setups has received relatively less attention. In~\cite{cao2017fundamental}, the optimal DoF of cache-aided MIMO networks with three transmitters and three receivers were studied, and in~\cite{cao2019treating}, general message sets were used to introduce inner and outer bounds on the achievable DoF of MIMO-CC schemes. More recently, we developed low-complexity MIMO-CC schemes for single-transmitter setups in~\cite{salehi2021MIMO}, and showed that with $G$ antennas at each receiver, if $\frac{L}{G}$ is an integer, the single-shot DoF of $Gt+L$ is achievable with a small subpacketization overhead. 
In~\cite{naseritehrani2024multicast, tehrani2024enhanced,naseritehrani2026cache,naseritehrani2023low}, we proposed an improved achievable single-shot DoF bound for MIMO-CC systems, going beyond the DoF value of $Gt+L$ earlier explored in~\cite{salehi2021MIMO}. In~\cite{naseritehrani2026cache}, we also designed a high-performance transmission scheduling strategy for MIMO-CC setups at finite SNR under symmetric per-user stream allocation. 
\color{black}Recently, vector coded caching in MU-MIMO with heterogeneous multi-antenna users was investigated in~\cite{zhao2026vector}, focusing on finite-SNR throughput under fading, pathloss, CSI overhead, and max--min fairness.
\color{black}
%


\subsubsection{Subpacketization bottleneck}
\label{rel_work_subs_4}
Subpacketization reflects the division of each file into smaller parts for the CC operation~\cite{lampiris2018adding}. Both the original single-antenna and MISO-CC schemes of~\cite{maddah2014fundamental,shariatpanahi2018physical} required exponentially growing subpacketization (w.r.t the user count $K$), rendering them infeasible for even moderate-sized networks~\cite{lampiris2018adding}. To resolve this issue, the pioneering works in~\cite{lampiris2018adding,salehi2020lowcomplexity} introduced signal-level CC operation (in contrast to bit-level CC~\cite{shariatpanahi2018physical}), showing that the same optimal DoF of $t + L$ can be achieved in MISO-CC setups with much smaller subpacketization. Later,~\cite{salehi2020lowcomplexity} was extended to MIMO setups~\cite{salehi2021MIMO} with linearly growing subpacketization. 
However, the reduced subpacketization in both schemes comes at the cost of limited applicability, as~\cite{lampiris2018adding} imposes divisibility constraints on the system parameters, requiring that both 
$\tfrac{L}{t}$ and $\tfrac{K}{t}$ are integers, and~\cite{salehi2020lowcomplexity} is applicable only to MISO-CC setups with $L \ge t$ (or $\big\lfloor \tfrac{L}{G} \big\rfloor \ge t$ for the MIMO-CC~\cite{salehi2021MIMO}). 
\color{black}
More recently, shared-cache MIMO~\cite{naseritehrani2026low} extended the
framework of~\cite{naseritehrani2026cache} to low-subpacketization asymptotic-DoF and finite-SNR designs, while MIMO-PDA schemes~\cite{huang2026placement} provided low subpacketization realizations of selected operating points with per-user spatial multiplexing gain fixed at $G$ and achievable DoF of $G(t+\lceil \frac{L}{G} \rceil)$. These constructions are subject to additional feasibility/divisibility conditions.
\color{black}

\subsection{Main Contributions}
\color{black}Existing MIMO-CC DoF results are predominantly developed for symmetric per-user receive-antenna configurations, whereas the heterogeneous setting considered here has received limited attention.\color{black} This setting fundamentally departs from symmetric baselines, as delivery must balance coded multicast opportunities against spatial bottlenecks induced by unequal antenna resources. We address this gap by jointly designing cache and spatial resources to leverage,  rather than be limited by, heterogeneity, thereby enhancing the achievable performance in asymmetric MIMO-CC systems.

We study the single-shot achievable DoF of asymmetric MIMO-CC with two objectives: 
(i) improving the achievable DoF under heterogeneous receive antenna configurations, and (ii) developing simple combinatorial and low-complexity designs
\color{black} whose achievable DoF is characterized either in closed form or through finite search over a discrete feasible set\footnote{\color{black}Throughout the paper, \emph{closed-form} refers to expressions that
require no search over design parameters, whereas expressions obtained by
optimization over a finite discrete feasible set are termed
\emph{finite-search optimized}.}.\color{black} 

Our developments are built under three baseline symmetric CC policies $\pi \in \{\mathrm{opt}, \mathrm{cmb}, \mathrm{lin}\}$, each specifying the cache placement and the delivery/multicasting design rule, while guaranteeing linear decodability for all target users. \color{black} 
Compared with the preliminary conference publication~\cite{naseritehrani2025achievable}, which was limited to \(\mathrm{opt}\)-based asymmetric constructions, the present article extends the framework to all three policies:
\color{black}(i) the DoF-optimized policy $\mathrm{opt}$~\cite{naseritehrani2026cache}, which attains
the best achievable single-shot DoF lower bound in symmetric MIMO-CC by optimizing the multicasting parameters (the number of served users and the number of streams per user);
\color{black}(ii) the combinatorial policy $\mathrm{cmb}$~\cite{shariatpanahi2016multi,salehi2021MIMO}, which views the MIMO-CC system as a set of parallel MISO-CC sub-systems, thereby avoiding
parameter search; and (iii) the linear-subpacketization policy $\mathrm{lin}$~\cite{salehi2020lowcomplexity,salehi2021MIMO}, which enables delivery via a
cyclic construction and offers a favorable
performance--complexity trade-off, approaching the DoF of the other policies with substantially lower subpacketization. We further 
provide a unified
linear-decodability condition used throughout the paper \color{black}(Remark~\ref{Linear-Decodability_general})
, identify the ceiling-based candidate as a special \(\mathrm{opt}\) operating point when its feasibility condition holds (Remark~\ref{opt_special_case_ceiling})
, quantify the achievable DoF gap analysis to
the \(\mathrm{opt}\) policy (Appendix~\ref{app_A0}), and clarify the extended applicability of $\rm lin$ (Remark~\ref{Flexible-stream cyclic operation}).\color{black}
%


These symmetric policies are adapted to the asymmetric receive-antenna setup through four asymmetric MIMO-CC strategies.
We first develop two primary strategies that explicitly reveal the fundamental tradeoff
between global caching and spatial multiplexing gains. The \emph{min-$G$} strategy enforces a symmetric setup by treating all users as having the same number of antennas equal to the minimum among them, thereby maximizing coded multicast opportunities at the cost of reduced spatial multiplexing. 
In contrast, the \emph{Grouping} strategy partitions users into subsets according to their receive antenna configurations and serves each subset in orthogonal (time) dimensions, thereby preserving per-group spatial multiplexing gains while sacrificing some global caching gain. 
To bridge this trade-off and unify the benefits of the two primary strategies, 
we then develop two hybrid strategies that extend the achievable DoF across a broader range of system parameters.
\color{black}Beyond~\cite{naseritehrani2025achievable}, we introduce the \emph{Super-grouping} strategy 
, which combines the primary strategies in a structured manner: users are first aggregated into equivalent sets using the \emph{min-$G$} abstraction, yielding equivalent Rx spatial multiplexing gains per set, and the equivalent sets are then served using the primary \emph{Grouping} strategy. \color{black}
Second, we introduce the \emph{Phantom} MIMO-CC strategy, which generalizes the primary strategies
by introducing virtual (``phantom'') antennas at some users to dynamically reconfigure spatial resources, enabling finer control over the number of served streams to accommodate asymmetric transmissions within each delivery interval.
%
\color{black}These hybrid strategies provide complementary operating points with distinct achievable DoF and subpacketization tradeoffs.
\emph{Super-grouping} closely approaches the performance of the \emph{Phantom} strategy 
by optimizing the number and composition of equivalent user sets, while
 offering a favorable complexity-performance trade-off through structured user partitioning based on the equivalent Rx-side spatial multiplexing. 
%
We also develop a simplified variant of
the 
\emph{Phantom} strategy with comparable achievable DoF and lower delivery complexity.

\color{black}
The
\emph{Phantom} strategy introduced in conference~\cite{naseritehrani2025achievable} is substantially generalized
from a single-round \(\mathrm{opt}\)-based scheme to a multi-round multicasting (MC) framework across the three policies. This extension includes round-dependent operating points and recursive constructions
, the general multi-round achievable DoF formulation (Theorem~\ref{thm:phantom_dof})
, and establishes linear decodability for all
policies, rounds, and transmitted intervals through a rank--nullity argument (Lemma~\ref{LD_Phantom}). In addition, the residual-unicasting (UC) phase is analyzed under a general UC completion-time formulation (Eqs.~\eqref{RA_UCT})\color{black}. 

\color{black}Across all proposed strategies under $\mathrm{opt}$, $\mathrm{cmb}$, and $\mathrm{lin}$ policies, we derive the achievable DoF and the associated delivery/subpacketization scalings either in closed-form or through finite-search optimization over discrete feasible sets, enabling a
transparent performance--complexity comparison. 
Numerical results illustrate 
that different policy--strategy combinations provide
flexible operating points for trade-offs
between achievable DoF and subpacketization.\color{black}

\emph{Notations.} Bold upper- and lower-case letters are used for matrices and vectors, respectively. Calligraphic letters denote sets, ${|\CK|}$ denotes set size of $\CK$, ${\CK \backslash \CT}$ represents elements in $\CK$ excluding those in ${\CT}$. \color{black}$(x)^+\triangleq\max\{x,0\}$ denotes the positive part of $x$. \color{black}
%
 \color{black}For better readability, Table~\ref{tab:notation} summarizes the main notation used throughout the paper.
\color{black}In the symmetric, \emph{min-\(G\)}, \emph{Grouping}, and \emph{Super-grouping} constructions, \(G\), \(\check G\), \(G_j\), and \(\bar G_l\), respectively, determine the
common per-user stream allocation. In \emph{Phantom} round \(i\), \(\hat G_i\) determines the hypothetical operating point, while projection onto the actual dimension \(\SfG_k\) gives the user-specific allocation.
\color{black}
\begin{table*}[t]
\color{black}
\caption{Main notations.}
\label{tab:notation}
\centering
\scriptsize
\renewcommand{\arraystretch}{1.10}
\setlength{\tabcolsep}{3pt}

\begin{tabularx}{\textwidth}{
|>{\centering\arraybackslash}p{0.105\textwidth}
 >{\raggedright\arraybackslash}X
|>{\centering\arraybackslash}p{0.105\textwidth}
 >{\raggedright\arraybackslash}X
|>{\centering\arraybackslash}p{0.105\textwidth}
 >{\raggedright\arraybackslash}X|}
\hline

\rule{0pt}{3.0ex}%
$(\CK,\ \CK_j,\ K_j,\bar\CK_l,\ \bar K_l)$
& \hspace*{2em} User set; group-$j$ set and size; super-group-$l$ set and size
&
$(\SfG_k,\ G_j, \bar G_l)$
& Rx dimensions of user $k$, group $j$, and super-group $l$
&
$(\pi)$
& Reference policy
\\

$(i,\ I,\ s)$
& MC round, round count, and interval index
&
$(\hat{\CK}_i,\ \hat K_i,\ \hat G_i)$
& Round-$i$ user set, size, and hypothetical Rx dimension
&
$(\hat{\Omega}_{\pi,i},\ \hat{\beta}_{\pi,i})$
& Candidate user and stream counts
\\

$(\bm\beta_{\pi,i}^{k}(s))$
& Actual stream allocation of user $k$
&
$(\hat{\CK}_{\pi,i}(s))$
& Served-user set in interval $s$
&
$(\tilde{\CM}_{\pi,i}^{k}(s),
\hat{\CM}_{\pi,i}^{k}(s))$
& \hspace*{2em} Candidate and transmitted subpacket sets
\\

$(\wp_\pi)$
& Policy-dependent packet index
&
$(\vartheta_\pi,\ \hat\varphi_{\pi,i},
\hat\Theta_{\pi,i})$
& Placement, round, and accumulated subpacketization
&
$(\tilde{\Bx}_{\pi,i}^{\rm MC}(s),
\hat{\Bx}_{\pi,i}^{\rm MC}(s))$
& \:\:  Candidate and actual MC vectors
\\

$(\CI_{\wp_\pi,k}(s),
\bar{\BH}_{\wp_\pi,k}(s))$
& \hspace*{2em} Interference set and equivalent channel
&
$(\Gamma_{\pi,i},
r_{\wp_\pi,k}(s))$
& Interfering-user and ZF-constraint counts
&
$(\delta_{\pi,i},
\delta_{\pi,i}^{k}(s))$
& Maximum and user-wise same-index subpacket counts
\\

$(\CK_{\rm UC},\ \CJ_{\rm UC},\ d_k)$
& Residual user/group sets and user-$k$ residual demand
&
$(Z_{\pi}^{\rm Tot},\ Z_{\pi}^{\rm UC})$
& Total missing and residual subpacket counts
&
$(S_{\pi,i}^{\rm MC},\ S_{\pi}^{\rm MC})$
& Round and total MC interval counts
\\

\rule[-1.15ex]{0pt}{3.0ex}%
$(S_{\pi,L}^{\rm UC})$
& BS-side UC completion time
&
$(S_{\pi}^{\rm UC})$
& UC completion time
&
$(\Theta_\pi)$
& Final subpacketization
\\

\hline
\end{tabularx}
\end{table*}
\color{black}
\section{System Model}
\label{section:sys_model}





We consider a MIMO setup like Figure~\ref{fig:ISIT_sysm}, where a single base station (BS) with \( L \) transmit antennas serves \( K \) cache-enabled multi-antenna users. Each user~$k \in [K]$ has $\SfG_k$ receive antennas,\footnote{In fact, \(L\) and 
\(\SfG_k\)
denote attainable spatial multiplexing gains, upper-bounded by the physical antenna counts, channel ranks, and RF chain limits; ``antenna count'' is used for simplicity.
}  
and a cache memory of size $MF$ data bits. The users request files from a library $\CF$ of $N$ files, each with the size of $F$ bits. Consequently, the cache ratio at each user is defined as $\gamma = \frac{M}{N}$. 
%
\begin{figure}[t]
        \centering
        \includegraphics[height = 4.3cm]{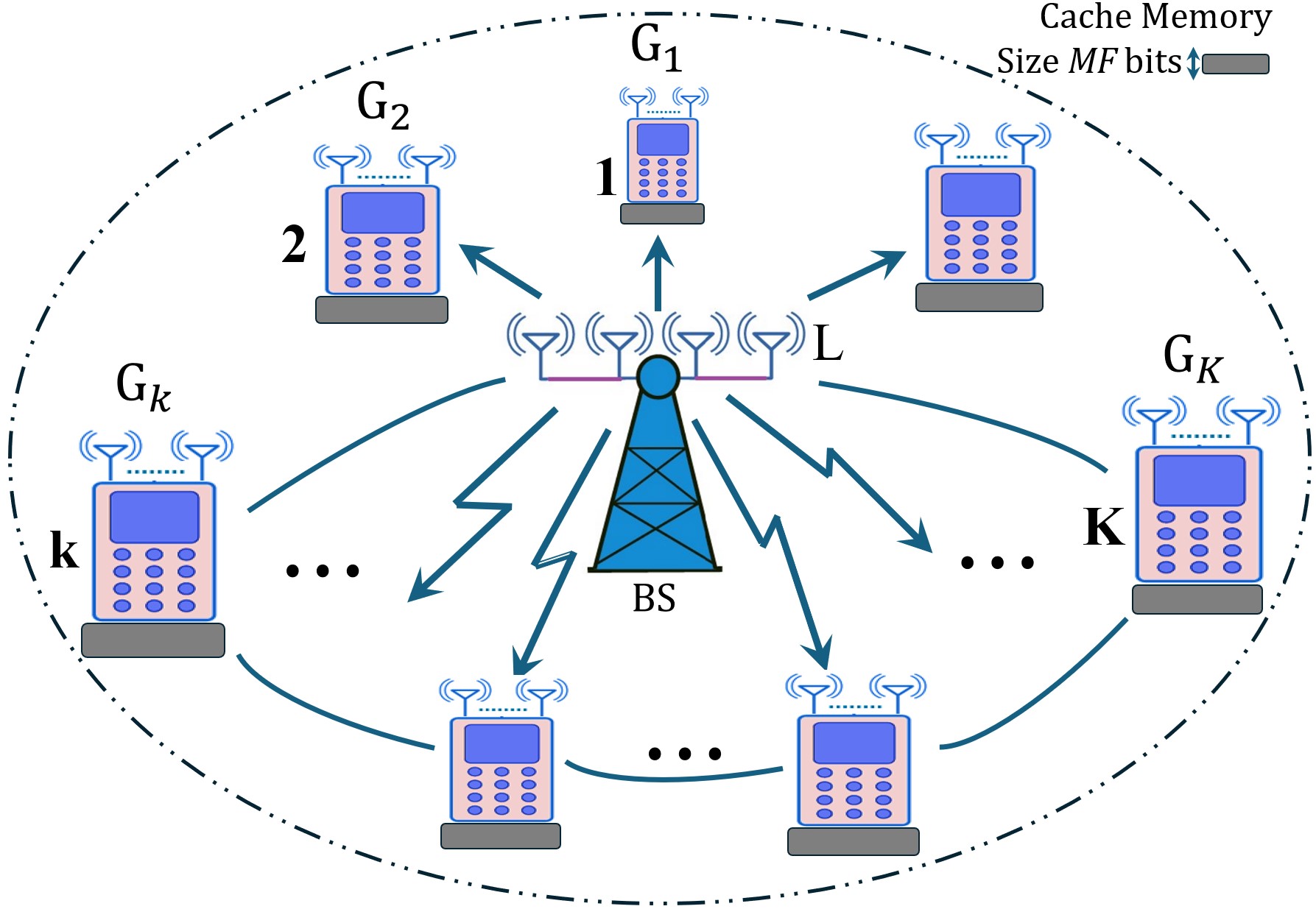} \label{fig:receiver_Prop}
    \caption{Asymmetric MIMO-CC downlink serving a target UE set \(\CK\) of size \(K\), where all UEs have the same cache size \(MF\) bits but heterogeneous Rx spatial multiplexing gains $\SfG_{k}$. 
    }\label{fig:ISIT_sysm}
\end{figure}
%
The user set $\CK \equiv [K]$ is partitioned into \( J \) disjoint subsets \(\CK_{j}\), $j \in \CJ = [J]$, such that each user ${k} \in$ \(\CK_{j}\) has \(\SfG_k = {G}_{j} \) receive antennas, i.e., ${G}_{j}$ indicates the number of receive antennas common for group $j$. We use $K_{j}$ to denote the size of $\CK_{j}$, and assume $\CK_{j}$ sets are sorted based on the antennas' count, i.e., ${G}_{1} < {G}_{2}$ $< \cdots < {G}_{J}$. 
\color{black}
Throughout the paper, we consider \(G_j\leq L\), \(j\in[J]\), so that the BS
has at least as many transmit dimensions as the largest receive dimension.%
\color{black}
The system operation consists of two phases: cache placement and delivery. In the placement phase, each file in the library is split into a number of smaller \textit{subfiles}, and the cache memory of each user is filled with a subset of subfiles, depending on the particular delivery scheme in use, as described in the following sections. 
At the beginning of the delivery phase, each user \(k\) reveals its requested file from the file library $\CF$ 
to the BS. 
The BS then constructs and transmits a set of transmission vectors \(\Bx(s) \in \mathbb{C}^L\), \(s \in [S]\), e.g., during $S$ consecutive time slots, where $S$ is given by the selected delivery scheme. Each transmission vector $\Bx(s)$ delivers a part of the requested data, with size $f(s)$ bits, to each user in a subset \(\CK(s) \subseteq \CK\) of target users, where \(|\CK(s)| = \Omega(s)\), and
\(\Omega(s)\) values are chosen to maximize a performance metric. 
Upon transmission of $\Bx(s)$, user $k \in \CK(s)$ receives
\begin{equation}
\label{eq:RX_signal}
\begin{array}{l}
\By_k(s) = \BH_k \Bx(s) + \Bz_k(s) \; ,
\end{array}
\end{equation}
where $\BH_k \in \mathbb{C}^{\SfG_k \times L}$ represents the channel matrix with rank $\min(\SfG_k,L)$ between the BS and user $k$, and $\Bz_k \sim \mathcal{CN}(\mathbf{0},N_0 \mathbf{I})$ is the noise. 
Full channel state information (CSI) is assumed to be available at the BS.\footnote{In practice, TDD enables BSs to estimate downlink channels from uplink pilots via channel reciprocity~\cite{tolli2019distributed}.}
Depending on the delivery algorithm, in each transmission, each subfile may be split into a different number of smaller \emph{subpackets} to ensure delivering new data in each transmission vector.

\textbf{Performance metric.}
As the placement phase is performed without any prior knowledge of the upcoming delivery phase, we can safely assume that each user has a portion $\gamma$ of its requested file in the cache, and requests the rest $F(1-\gamma)$ bits from the BS.
In general, the sum rate for delivering the missing $KF(1-\gamma)$ bits to all $K$ users can be defined as
$R_{\mathrm{total}} = \frac{K F (1-\gamma)}{T_{\mathrm{total}}}$
bits per second (bps), where $T_{\mathrm{total}}$ is the total delivery time. 
Let us use
$\bm\beta_{k}(s)$ to denote the number of received parallel streams, and $R_{k,l}(s)$ to denote the achievable rate of stream $l \in [\bm\beta_k(s)]$ of user $k$ in transmission $s$. Then, denoting the delivery time for transmission $s$ as $T(s)$, we can write
\begin{equation}
    \begin{array}{l}
    T_{\mathrm{total}} = \!\sum_{s \in [S]} T(s) = \!\sum_{s \in [S]} \frac{\sum_{k \in \CK(s)} \bm\beta_k(s) f(s)}{\sum_{k \in \CK(s)} \sum_{l=1}^{\bm\beta_k(s)} R_{k,l}(s)}.
    \end{array}
\end{equation}
Clearly, $R_{k,l}(s)$ (and hence, $T_{\mathrm{total}}$ and $R_{\mathrm{total}}$) is a function of the transmission SNR.
%
In this paper, we are interested in characterizing the achievable spatial DoF, defined as:
\begin{equation}\label{eq:DoF_ideal}
\begin{array}{l}
    \mathrm{DoF} = \lim\limits_{\textrm{SNR}\rightarrow\infty}\frac{{R}_{\textrm{total}}}{\log \textrm{SNR}} = \frac{KF(1-\gamma)}{\color{black}\lim\limits_{\textrm{SNR}\rightarrow\infty} (\log \textrm{SNR} \times T_{\textrm{total}}\color{black})
    }.
\end{array}
\end{equation}
Assuming that inter-stream interference can be completely removed by zero-forcing~(ZF) precoders, as $\textrm{SNR} \to \infty$, each stream can be delivered at a rate approaching the channel capacity, i.e., 
\begin{equation}
\begin{array}{l}
    \lim\limits_{\textrm{SNR} \to \infty} R_{k,l}(s) = C \coloneq \lim\limits_{\textrm{SNR}\rightarrow\infty} \log(\textrm{SNR}), \\[1ex]  \qquad  \qquad \qquad \forall s \in [S], k \in \CK(s), l \in [\bm\beta_k(s)].
\end{array}
\end{equation}
Hence, we get $\lim_{\textrm{SNR}\rightarrow\infty} T_{\textrm{total}} = \sum_{s \in [S]} \frac{f(s)}{C}$, which together with~\eqref{eq:DoF_ideal} results in:
\begin{equation}
    \begin{array}{l}\label{eq:DoF_ideal1}
    \mathrm{DoF} = \frac{K F(1-\gamma)}{\sum_{s\in [S]}f(s)}.
    \end{array}
\end{equation}
Now, as each user $k \in \CK(s)$ receives $\bm\beta_k(s)$ parallel streams in transmission $s$, and all the missing subpackets of users are delivered after all $S$ transmissions, we have $KF (1-\gamma) = \sum_{s \in [S]} f(s) 
\sum_{k \in \CK(s)} \bm\beta_k(s)$, and~\eqref{eq:DoF_ideal1} can be re-written as
\begin{equation}\label{eq:asym_DoF_phantom_version}
\begin{array}{l}
    \mathrm{DoF} = \frac{\sum_{s \in [S]} f(s) \cdot \sum_{k \in \CK(s)}\bm\beta_k(s)}{\sum_{s \in [S]} f(s)}.
\end{array}
\end{equation}
Depending on the particular delivery scheme, all the users may receive the same number of streams per transmission, i.e., $\bm\beta_k(s) = \beta(s)$ for all $k \in \CK(s)$. In this case, equation~\eqref{eq:asym_DoF_phantom_version}
can be further simplified to
\begin{equation}\label{eq:asym_DoF}
\begin{array}{l}
    \mathrm{DoF} = \frac{\sum_{s \in [S]} f(s) \cdot \beta(s) \cdot \Omega(s)}{\sum_{s \in [S]} f(s)},
\end{array}
\end{equation}
which can also be interpreted as the weighted sum (by the factor of subpacket length $f(s)$) of per-transmission DoF values (total number of parallel streams, i.e., $\beta(s)\Omega(s)$).
\begin{remarknum}\label{standard_dof_analysis}
The DoF analysis assumes perfect CSI and ideal interference suppression and characterizes the number of independent streams that can be delivered in parallel in the high-SNR regime. Each DoF operating point therefore specifies a feasible number of simultaneously served users and streams per user, providing guidance for transmission-mode selection under practical constraints. With imperfect CSI or nonideal interference suppression, some spatial dimensions may need to be reserved for robustness, while finite-SNR operation may also favor serving fewer simultaneously served users and/or streams~\cite{tolli2017multi}. In particular, if the CSI uncertainty does not vanish as the SNR increases, the resulting residual interference may prevent the nominal DoF from being achieved.
\end{remarknum}
%
\section{Reference Symmetric MIMO-CC Policies}
\label{ref-schemes}
In this paper, we propose multiple transmission strategies for asymmetric (in terms of the number of receive antennas) MIMO-CC systems. Each of these strategies can itself be built upon various reference
symmetric MIMO-CC schemes, herein called \emph{policies}, available in the literature\cite{naseritehrani2026cache,shariatpanahi2016multi,tolli2017multi,salehi2020lowcomplexity,salehi2021MIMO}. Here, we consider three reference 
policies, represented by $\pi\in\{\mathrm{opt},\mathrm{cmb},\mathrm{lin}\}$. Depending on the chosen policy, the resulting asymmetric design naturally leads to different tradeoffs, including different achievable DoF levels and varying subpacketization requirements. 
%
\color{black} 
%
\color{black}
Throughout this section, $t \triangleq K\gamma$ denotes the coded-caching gain. Whenever $t$ is noninteger, the corresponding performance 
can be extended via memory sharing between adjacent integer values of~$t$.
%
%
%
\color{black}

\subsection{The DoF-optimized policy $(\mathrm{opt})$}
%

To the best of our knowledge, \color{black}the largest known linear achievable single-shot DoF lower bound \color{black} for any general MIMO-CC setup with a symmetric number of $G$ antennas at all users is established in~\cite{naseritehrani2026cache}. 

\vspace{1pt}
\noindent\textbf{Placement phase.} Each file $W \in \CF$ is split into ${\vartheta}_{\mathrm{opt}} = \binom{K}{K \gamma}$ subfiles $W_{\CP}$, where $\CP$ can denote every subset of the users with size $K\gamma$. Then, at the cache memory of each user~$k$, we store all $W_{\CP}$ for which $ k \in \CP$.

\vspace{1pt}
\noindent\textbf{Delivery phase.}
The algorithm presented in~\cite{naseritehrani2026cache} assumes the same number of $\Omega \ge K\gamma +1$ users are served in each transmission, and each user receives the same number of parallel streams $\beta$. $\Omega$ and $\beta$ are found by solving the DoF optimization problem\footnote{A simple one-dimensional line search over the integer $\Omega$ is obtained by enforcing the $\beta$-constraint with equality (i.e., the floor term), yielding~\cite{naseritehrani2026cache}

$\Omega^* =  \arg\max\limits_{\substack{\\[.3ex] K\gamma+1 \le \Omega \le K\gamma+L}} \Omega  \bigg\lfloor \min\Big(G, \frac{ L \binom{\Omega-1}{K\gamma}}{1 + (\Omega - K\gamma-1)\binom{\Omega-1}{K\gamma}}\Big)\bigg\rfloor$.} 
\begin{equation}\label{eq:total_DoF0}
\begin{array}{l}
 (\Omega^*, \beta^*) = \Delta(K, \gamma, L, G),
 \end{array}
\end{equation}
where
\begin{equation}\label{eq:total_DoF}
\begin{array}{l}
   \Delta(K, \gamma, L, G) \triangleq  \arg \max_{\Omega,\beta}~ \Omega \cdot \beta, \\[1ex]
    \mathrm{s.t.}\:\:{\beta \le \mathrm{min}\bigg({G}, \frac{L \binom{\Omega-1}{K\cdot \gamma}}{1 + (\Omega - K\cdot \gamma-1)\binom{\Omega-1}{K\cdot \gamma}} \bigg).}
\end{array}
 \end{equation}

In~\cite[Theorem~2]{naseritehrani2026cache}, it is shown that for any general $\Omega$ and $\beta$ values satisfying the condition in~\eqref{eq:total_DoF}, one can always design a MIMO-CC scheme where in each transmission, every target user can decode all its $\beta$ parallel streams with a linear receiver. 
The algorithm involves splitting each subfile $W_{\CP}$ into ${\phi}_{\mathrm{opt}} = \binom{K-K\gamma-1}{\Omega^*-K\gamma-1} \beta^*$ subpackets $W_{\CP}^q$, and creating 
\begin{equation}\label{eq:tran_cnt_minG1}
 \begin{array}{l}
 {S}_{\mathrm{opt}} = \binom{K}{\Omega^*} \binom{\Omega^*-1}{K\gamma}
\end{array}
\end{equation}
transmission vectors, each delivering $\beta^*$ parallel streams to every user within a target subset of $\Omega^*$ users. Accordingly, $\lambda_{\mathrm{opt}}=\binom{K-1}{\Omega^*-1}\binom{\Omega^*-1}{K\gamma}$ is the total number of
MC transmissions that serve a given user $k\in{\CK}$.

\vspace{1pt}
\noindent\textbf{DoF analysis.} 
The DoF of the $\mathrm{opt}$ policy is given as
\begin{equation}\label{eq:DoF_minG}
 \begin{array}{l}
 \textrm{DoF}^{*}_{\mathrm{opt}} = \Omega^* \cdot \beta^*.
\end{array}
\end{equation}
This follows from~\eqref{eq:asym_DoF}, by allocating an equal size of 
\begin{equation}\label{files_size_sym}
\begin{array}{l}
f = \tfrac{F}{\Theta_{\mathrm{opt}}}\ \text{bits}    
\end{array}
\end{equation}
per all streams, where 
\begin{equation}\label{sym_subpkt}
\begin{array}{l}
    \Theta_{\mathrm{opt}} = {\vartheta}_{\mathrm{opt}} {\phi}_{\mathrm{opt}} = \binom{K}{K\gamma} \binom{K-K\gamma-1}{\Omega^{*}-K\gamma-1}\beta^{*},
    \end{array}
\end{equation} 
denotes the final subpacketization. In DoF calculations, design parameters are set to their optimized values, i.e., $\Omega(s) = \Omega^*$ and $\beta(s) = \beta^*$, for every $s \in [{S}_{\mathrm{opt}}]$.
\color{black}
Within the candidate \((\Omega,\beta)\) operating points of the \(\mathrm{opt}\) framework, the choice \(\beta=G\) yields a simple ceiling-based benchmark. Its feasibility and relation to the general $\mathrm{opt}$ policy are characterized next.
\begin{remarknum}\label{opt_special_case_ceiling}
The search-based \(\mathrm{opt}\) policy includes the ceiling-based achievable DoF as a particular solution whenever the corresponding operating point satisfies the linear-decodability condition.
According to~\cite[Corollary~1]{naseritehrani2026cache},
the \(\mathrm{opt}\) policy considers feasible pairs \((\Omega,\beta)\)
satisfying $\beta\leq G$ and
\begin{equation}
\Big\lceil
\nicefrac{\beta}{\binom{\Omega-1}{t}}
\Big\rceil
\leq
L-(\Omega-t-1)\beta,
\label{eq:opt_ceiling_feasibility}
\end{equation}
where the corresponding achievable DoF is \(\Omega\beta\). Now, consider the parameter choice
$\beta=G$, $\Omega=t+\lceil \frac{L}{G}\rceil$.
For this choice,~\eqref{eq:opt_ceiling_feasibility} reduces to
\begin{equation}
\Big\lceil
\nicefrac{G}{\binom{t+\lceil \frac{L}{G}\rceil-1}{t}}
\Big\rceil
\leq
L-(\lceil \nicefrac{L}{G}\rceil-1)G.
\label{eq:ceiling_special_feasibility}
\end{equation}
Whenever~\eqref{eq:ceiling_special_feasibility} holds, the corresponding achievable DoF is
\begin{equation}\label{ceil_DoF}
{\rm DoF}_{\rm ceil} = \Omega\beta
=
G\big(t+\big\lceil\nicefrac{L}{G}\big\rceil\big).
\end{equation}
Consequently,
$\mathrm{DoF}_{\rm opt}
\geq
{\rm DoF}_{\rm ceil}$. 
To quantify the potential gain in achievable DoF from optimizing
\((\Omega,\beta)\) beyond this feasible ceiling benchmark,
Appendix~\ref{app_A0} upper-bounds the resulting DoF gap.
Moreover, substituting the ceiling operating point into the \(\mathrm{opt}\) subpacketization and transmission-count expressions gives
$\Theta_{\rm ceil}
\!=\!
\binom{K}{t}
\binom{K-t-1}{\left\lceil \nicefrac{L}{G}\right\rceil-1}
G$,
and
$S_{\rm ceil}
\!=\!
\binom{K}{t+\left\lceil \nicefrac{L}{G}\right\rceil}
\binom{t+\left\lceil \nicefrac{L}{G}\right\rceil-1}{t}$.
Therefore, whenever~\eqref{eq:ceiling_special_feasibility} holds, the
ceiling-based construction is included in the feasible design space of the more
general \(\mathrm{opt}\) policy\footnote{\color{black}
For \(\beta\!=\!G\), Lemma~1 in~\cite{naseritehrani2026cache}
provides only the necessary bound
\(\Omega\!\leq\! t+\lfloor L/G\rfloor+1\), whereas
\eqref{eq:opt_ceiling_feasibility}, rewritten as $(\Omega-t-1)
+\nicefrac{1}{\binom{\Omega-1}{t}}
\leq\nicefrac{L}{G}$, yields a sharper bound
\(\Omega\!\leq\! t+\lceil L/G\rceil\).
The bounds coincide when \(L/G\!\notin\!\mathbb{Z}\). 
When \(L/G\!\in\!\mathbb Z\), the upper bound from~\cite[Lemma~1]{naseritehrani2026cache} exceeds the sharper bound by one.
Neither bound alone guarantees feasibility, which additionally
requires~\eqref{eq:ceiling_special_feasibility}.}. This result also coincides with~\cite[Theorem~3]{huang2026placement}.
\end{remarknum}
\color{black}

\subsection{The combinatorial parallel links policy $(\mathrm{cmb})$}

This policy builds on the MISO-CC scheme in~\cite{shariatpanahi2016multi,tolli2017multi} and adopts key delivery-design insights from the MIMO-CC framework in~\cite{salehi2021MIMO}. In particular, the $\mathrm{cmb}$ policy views the system as a set of $G$ parallel MISO-CC sub-systems, and uses receiver-side spatial multiplexing to suppress (parts of) the interference these sub-systems cause on each other. The main motivation for considering this policy is that it provides a closed-form DoF expression 
and requires no line search over the multicasting parameters, unlike the \(\mathrm{opt}\) policy.

\vspace{1pt}
\noindent\textbf{Placement phase.}
Similar to the placement phase of the $\mathrm{opt}$ policy, each file $W \in \CF$ is split into ${\vartheta}_{\mathrm{cmb}} = \binom{K}{K \gamma}$ subfiles $W_{\CP}$, and each user $k$ stores all $W_{\CP}$ for which $k \in \CP$.

\vspace{1pt}
\noindent\textbf{Delivery phase.}
The first step is to consider a \emph{virtual} MISO-CC network, with 
the transmitter-side spatial multiplexing gain $\alpha' \equiv  \lfloor \frac{L}{G} \rfloor$. Then, we use the MISO-CC scheme in~\cite{shariatpanahi2016multi,tolli2017multi} to create a set of $S'=\binom{K}{t+\alpha'}$ \emph{virtual} transmission vectors (these vectors are not transmitted and serve merely as placeholders to create the final vectors). Following~\cite{shariatpanahi2016multi,tolli2017multi}, this requires extra subpacketization by a factor of $\phi' = \binom{K-t-1}{\alpha'-1}$, and $t+\alpha'$ users are included in each virtual transmission vector. However, directly using the resulting virtual vectors could necessitate the target users to employ non-linear receivers (if the number of parallel streams per user exceeds their multiplexing capacity).
To alleviate this issue, we replace XOR terms with signal-level processing, and replace each virtual transmission vector with $\binom{t+\alpha'-1}{t}$ new virtual vectors each delivering exactly one term to each of the $t+\alpha'$ target users (see~\cite{salehi2022multi}).
Now, to create the final transmission vectors, we repeat each resulting virtual vector $G$ times (by further splitting every subpacket by a factor of $G$ and using new subpacket indices in each repetition), and transmit all the $G$ repetitions of a virtual vector with a single transmission vector. The receiver-side (as well the available transmitter-side) spatial multiplexing capacity is used to remove the extra inter-stream interference caused by such parallel transmissions. 
Following this algorithm, the final subpacketization level is
\begin{equation}\label{cpar_sbpkt}
\begin{array}{l}
    \Theta_{\mathrm{cmb}} = G\binom{K}{K \gamma}\binom{K-K\gamma-1}{\lfloor L/G \rfloor -1}
     \end{array}
\end{equation}
and the number of transmission slots is
\begin{equation}\label{cpar_txcnt}
\begin{array}{l}
S_{\mathrm{cmb}} = \binom{K}{K\gamma+\lfloor L/G \rfloor}\binom{K\gamma+\lfloor L/G \rfloor -1}{K\gamma}.
\end{array}
\end{equation}
Also, $\lambda_{\mathrm{cmb}}=\binom{K-1}{\lfloor\frac{L}{G}\rfloor+K\gamma-1}\binom{\lfloor\frac{L}{G}\rfloor+K\gamma-1}{K\gamma}$
gives the number of MC transmissions that include a given user $k\in{\CK}$.

\vspace{1pt}
\noindent\textbf{DoF analysis.}
Each user now receives $\beta(s)= \beta={G}$ parallel streams per transmission, while the number of users per transmission is fixed at $\Omega(s) = \Omega=t+\alpha'$. As a result: 
\begin{equation}\label{eq:DoF_minG_cpar}
\begin{array}{l}
 \mathrm{DoF}_{\mathrm{cmb}}=\Omega\cdot\beta = G(t+\alpha') = G K\gamma + G\lfloor\frac{L}{G}\rfloor.
 \end{array}
\end{equation}
%
%
Similar to the $\mathrm{opt}$ policy, the $\mathrm{cmb}$ policy is also applicable to any set of network parameters ($K$, $\gamma$, $L$, $G$). However, both policies suffer from a large subpacketization overhead due to their underlying combinatorial structures. 
The core motivation behind considering the $\mathrm{cmb}$ policy is its closed-form DoF expression, eliminating the line search to determine the optimized operating point. \color{black}
The $\mathrm{cmb}$-based construction guarantees linear decodability, as discussed in Remark~\ref{Linear-Decodability_general}. The DoF gap between $\mathrm{opt}$ and $\mathrm{cmb}$ is at most $2G-2$~\cite{naseritehrani2026cache}. \color{black} 
When the operating point in Remark~\ref{opt_special_case_ceiling} is feasible, the refined Dof gap is at most $(G-3)^+$ for $L/G\in \mathbb Z_{+}$, and lies in $[G,G+(G-3)^+]$ for $L/G\notin \mathbb Z_{+}$
(Appendix~\ref{app_A0}). These bounds also apply to
the $\mathrm{lin}$ policy within its feasibility condition.
%
%
\color{black}
\subsection{The cyclic parallel links policy $(\mathrm{lin})$}\label{cyc_subsubs}

This policy addresses the large subpacketization bottleneck of the other two policies, 
but is applicable under the feasibility condition of $\lfloor L/G\rfloor \ge K\gamma$. 
The idea is to follow a similar process as $\mathrm{cmb}$ policy, but to apply the cyclic MISO-CC scheme in~\cite{salehi2020lowcomplexity} to create virtual transmission vectors for the virtual MISO network (here, there is no need to increase the number of transmission vectors as the cyclic scheme is inherently linear). 

\vspace{1pt}
\noindent\textbf{Placement phase.}
Data placement is done using a cyclic structure, as described in~\cite{salehi2020lowcomplexity}. Each file is split into ${\vartheta}_{\mathrm{lin}} = K$ packets $W_p$, $p \in [K]$. Then, user~1 stores $W_K,W_1,\cdots,W_{K\gamma-1}$, and for each user $k$, $k>1$, the set of stored packet indices is found by a circular shift of the indices for user $k-1$, over the index set $[K]$ (see~\cite{salehi2020lowcomplexity} for details and graphical representations). 

\vspace{1pt}
\noindent\textbf{Delivery phase.}
We again treat the asymmetric MIMO-CC system as a set of ${G}$ parallel MISO-CC subsystems. Similar to the 
$\mathrm{cmb}$ policy, we consider a virtual MISO-CC scheme with 
$\alpha' \equiv \lfloor \tfrac{L}{{G}} \rfloor$. However, now, we apply the cyclic scheme in~\cite{salehi2020lowcomplexity} to generate the virtual transmission vectors. We split each packet $W_p$ into $\phi^{}_{\mathrm{lin}} = t + \alpha'$ subpackets $W_p^q$, and create $K(K-t)$ virtual transmission vectors through
two perpendicular cyclic shift operations over a tabular structure~\cite{salehi2020lowcomplexity} 
(details are not repeated here for the 
sake of
brevity). Then, we create a final transmission vector for each virtual vector, by repeating the latter by ${G}$ times and transmitting all the copies in parallel.
The corresponding required subpacketization is now
\begin{equation}\label{cycl_sbpkt}
\begin{array}{l}
    \Theta_{\mathrm{lin}} = GK(K\gamma + \lfloor\frac{L}{G}\rfloor)
    \end{array}
\end{equation}
and the number of transmissions is
\begin{equation}\label{cycl_tx_cnts}
 S_{\mathrm{lin}} = K(K-K\gamma).   
\end{equation}
Accordingly $\lambda_{\mathrm{lin}}= (K-K\gamma)(K\gamma+\lfloor\frac{L}{G}\rfloor)$ denotes the total number of transmissions that serve a given user $k\in {\CK}$.

\vspace{1pt}
\noindent\textbf{DoF analysis.}
Provided that the condition $\lfloor \frac{L}{G} \rfloor \ge K \gamma$~\cite{salehi2021MIMO} holds,
the DoF of $\mathrm{lin}$ policy remains the same as in the $\mathrm{cmb}$ policy, as $\beta(s)= \beta={G}$, and $\Omega(s)= \Omega={t+\alpha'}$:
\begin{equation}\label{cycl_DoF}
\begin{array}{l}
    \mathrm{DoF}_{\mathrm{lin}} = \mathrm{DoF}_{\mathrm{cmb}}=\Omega\cdot\beta =   G K\gamma + G\lfloor\frac{L}{G}\rfloor.
    \end{array}
\end{equation}
\vspace{-5mm}
\color{black}\begin{remarknum}
\label{Flexible-stream cyclic operation}
If $\lfloor \frac{L}{G} \rfloor< K \gamma$, the $\mathrm{lin}$
policy remains applicable by reducing the
number of parallel streams to $\beta<G$
such that $\lfloor \frac{L}{\beta} \rfloor \geq K \gamma$.
The corresponding DoF and subpacketization
follow by replacing $G$ with $\beta$ in the
cyclic construction, thus extending its
applicability while retaining linear
subpacketization for a larger set of network parameters. 
In the rest of the paper, to avoid confusion with the notations, we use $G$ as the number of parallel streams, and only in simulation results in Section~\ref{section:Simulations}, we limit the number of parallel streams if $\lfloor \frac{L}{G} \rfloor < K\gamma$.
\end{remarknum}\color{black}
\vspace{-3mm}
%
\begin{remarknum}\label{Linear-Decodability_general}
The linear decodability of the $\mathrm{opt}$ policy is proven in~\cite[Theorem 1]{naseritehrani2026cache}. The same reasoning used therein can also be employed to show linear decodability of both $\mathrm{cmb}$ and $\mathrm{lin}$ policies. In general, if we use $\delta_{\pi}$ to denote the maximum number of subpackets with the same packet index in each transmission interval with policy $\pi$, the proof in~\cite[Theorem 1]{naseritehrani2026cache} implies that linear decodability requires
\begin{equation} \label{DoF_bndd}
\begin{aligned}
    \beta &\le \min \left(G, \frac{L-\delta_{\pi}}{\Omega-K\gamma-1} \right), \qquad \Omega > K\gamma+1, \\[-.95mm]
    \beta &\le G, \qquad \qquad \qquad \qquad \qquad \quad \; \Omega = K\gamma+1.
\end{aligned}
\end{equation}
\color{black}
The linear decodability of the \(\mathrm{cmb}\) construction follows by setting 
\(\delta_{\pi}\!=\!\beta\!=\!G\)
and \(\Omega\!=\!K\gamma+\lfloor L/G\rfloor\) in~\eqref{DoF_bndd}.
The same holds for the $\mathrm{lin}$ construction under its
feasibility condition $\lfloor L/G\rfloor\!\geq \!K\gamma$\footnote{\color{black} 
The condition also holds for $\rm lin$ policy 
when 
the number of parallel streams are reduced to $\beta < G$, thereby extending the feasibility region, as discussed in Remark~\ref{Flexible-stream cyclic operation}. \color{black}}.
For the ceiling-based operating point, setting $\beta\!=\!G$, \(\delta_{\rm ceil}\!=\! \Big\lceil \frac{G}{\binom{K\gamma+\lceil L/G\rceil-1}{K\gamma}}
\Big\rceil\) and \(\Omega\!=\!K\gamma+\lceil L/G\rceil\) in~\eqref{DoF_bndd} recovers exactly the feasibility condition in~\eqref{eq:ceiling_special_feasibility}. Thus, unlike the guaranteed-feasible $\mathrm{cmb}$ point, the ceiling-based point is further feasibility-constrained because, for some parameter combinations, serving the additional user may leave insufficient spatial dimensions for linear decodability.
\end{remarknum}

\section{Primary asymmetric solutions}
In this section, we analyze two distinct canonical CC transmission strategies for asymmetric MIMO-CC systems. The \emph{min-$G$} targets to maximize the global CC gain, while \emph{Grouping} aims at maximizing the receiver-side spatial multiplexing gain.
\subsection{ \texorpdfstring{The min-$G$ strategy}{The min-$G$ strategy}}
%
\begin{table*}[!t]
    \centering
    \caption{The \emph{min-$G$} strategy parameters for different reference policies}
    \label{tab:minG-table}
    \begin{tabular}{>{\centering\arraybackslash}p{0.12\textwidth}!{\vrule width 1pt}>{\centering\arraybackslash}p{0.26\textwidth}|>{\centering\arraybackslash}p{0.26\textwidth}|>{\centering\arraybackslash}p{0.26\textwidth}}
        \textbf{Policy} & \textbf{DoF-optimized} & \textbf{Combinatorial parallel links} & \textbf{Cyclic parallel links} \\
        \Xhline{1pt}
        \textbf{DoF} 
        & $\begin{aligned}\\[-3.mm]
        \textrm{DoF}^{*}_{\mathrm{opt},\check{G}} &= \Omega^*_{\check{G}} \cdot \beta^*_{\check{G}} \\
        \Omega^*_{\check{G}} , \beta^*_{\check{G}} &\leftarrow \Delta(K,\gamma,L,\check{G})
        \end{aligned}$
        & \vspace{-2mm}$\mathrm{DoF}^{}_{\mathrm{cmb},\check{G}} = \check{G}\big(\color{black}K\gamma + \big\lfloor \tfrac{L}{\check{G}} \big\rfloor\big)\color{black}$ 
        & $\mathrm{DoF}^{}_{\mathrm{lin},\check{G}} = \check{G}\big(\color{black}K\gamma +  \big\lfloor \tfrac{L}{\check{G}} \big\rfloor\big)\color{black}$ \\
        \hline
        \textbf{Subpacketization} 
        & \vspace{.2mm}$\Theta^{}_{\mathrm{opt},\check{G}} = \beta_{\check{G}}^*\binom{K}{K\gamma}\binom{K-K\gamma-1}{\Omega_{\check{G}}^*-K\gamma-1}$ 
        & \vspace{.2mm}$\Theta^{}_{\mathrm{cmb},\check{G}} = \check{G} \binom{K}{K \gamma}\binom{K-K\gamma-1}{\big\lfloor \tfrac{L}{\check{G}} \big\rfloor -1}$ 
        & \vspace{.2mm}$\Theta^{}_{\mathrm{lin},\check{G}} = \check{G} K \left( K\gamma + \big\lfloor \tfrac{L}{\check{G}} \big\rfloor \right)$ \\
        \hline
        \textbf{Trans. count} 
        & \vspace{.1mm}${S}_{\mathrm{opt},\check{G}} = \binom{K}{\Omega^*_{\check{G}}} \binom{\Omega^*_{\check{G}}-1}{K\gamma}$ 
        & \vspace{.1mm}$S^{}_{\mathrm{cmb},\check{G}} = \binom{K}{K\gamma+\big\lfloor \tfrac{L}{\check{G}} \big\rfloor}\binom{K\gamma+\big\lfloor \tfrac{L}{\check{G}} \big\rfloor-1}{K\gamma}$ 
        & \vspace{.1mm}$S^{}_{\mathrm{lin},\check{G}} = K\left(K-K\gamma\right)$ \\
        \hline
        \textbf{Applicability} 
        & All network parameters 
        & All network parameters 
        & $\big\lfloor \tfrac{L}{\check{G}} \big\rfloor \ge K\gamma$ \\
    \end{tabular}%
\end{table*}
\begin{figure}[t]
         \centering
        \includegraphics[height = 1.94cm]{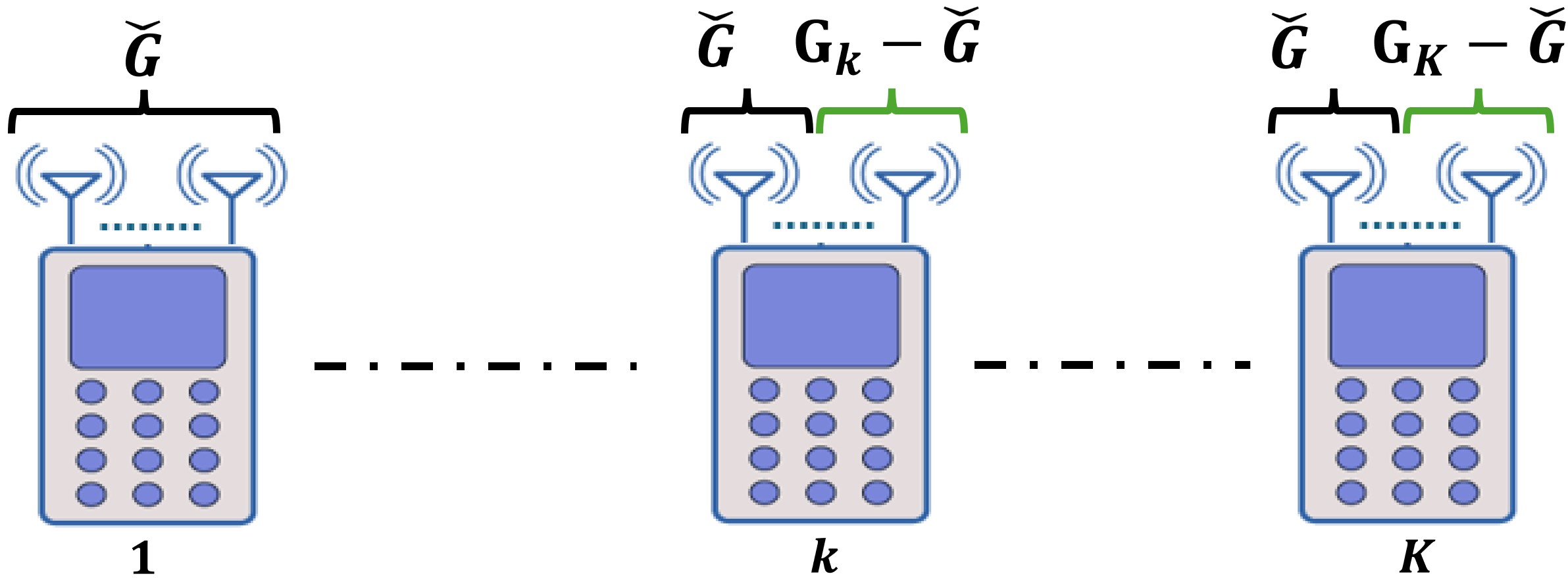} 
    \caption{The \emph{min-$G$} delivery strategy. Spatial multiplexing capacity is fixed to $\check{G}$ for all $k\in [K]$}\label{fig:min_G}
\end{figure}
The \emph{min-$G$} strategy transforms the asymmetric MIMO-CC system into a symmetric setup by considering the effective spatial multiplexing gain of
\begin{equation}\label{eq:min_G}
 \begin{array}{l}
 \check{G}  = \min_{k\in\CK} \SfG_k = G_{1}
\end{array}
\end{equation}
for all users. The idea behind this strategy is to maximize the effect of the cumulative cache size of users (i.e., the CC gain) in the CC delivery session. However, for users in groups $\CK_{j}$, $j \ge 2$, a part of the spatial multiplexing capability ($G_{j}-\check{G}$) is left unused (see Fig.~\ref{fig:min_G}).
Built upon this core idea, the analytical results of the \emph{min-$G$} strategy for various reference policies, obtained by substituting $G$ with $\check{G}$ in respective equations, are summarized in Table~\ref{tab:minG-table}.
\subsection{The Grouping strategy}
\label{subsec: grouping}
With the \emph{Grouping} strategy, we treat each subset $\CK_{j}$ of users separately (see Fig~\ref{fig:grps}). With this approach, all users in a set $\CK_{j}$ apply the same reference policy locally, independent of users in other subsets. 
The primary rationale behind this strategy is to maximize the spatial multiplexing gain for each subset $\CK_{j}$, at the cost of reducing the effective cumulative cache size (i.e., the CC gain). Built upon this core idea, all the analytical results developed for the 
proposed
\emph{Grouping} strategy are summarized in Table~\ref{tab:grouping-table}.

\vspace{1pt}
\noindent\textbf{Placement phase.}
Data placement is performed in $J$ steps, where at step $j \in \CJ$, cache memories of users in $\CK_{j}$ are filled with data, following the selected reference policy, by replacing $K$ with $K_{j}$.

\vspace{1pt}
\noindent\textbf{Delivery phase.} Data delivery is also done in $J$ orthogonal steps, following the selected reference policy, by replacing $K$ with $K_{j}$ and $G$ with $G_{j}$. At step $j \in \CJ$, all the missing parts of data files requested by users in $\CK_{j}$ are delivered.

\vspace{1pt}
\noindent\textbf{DoF analysis.}
The following lemmas clarify how the DoF expressions in Table~\ref{tab:grouping-table} are calculated.

\begin{figure} [!t]
         \centering
        \includegraphics[height = 2.8cm]{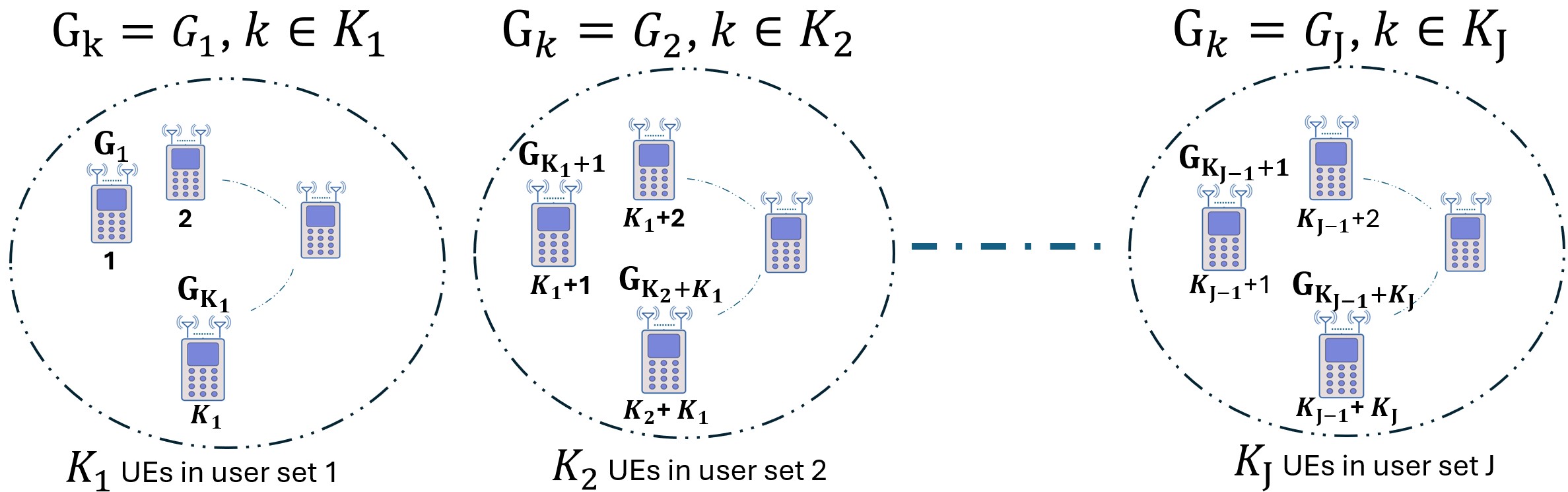} 
    \caption{The \emph{Grouping}  strategy with the spatial multiplexing capacity of $G_{j}$ and $K_{j}$ users in group $j\in\CJ$}\label{fig:grps}
\end{figure}

\begin{table*}[!t]
    \centering
    \caption{The \emph{Grouping} strategy parameters for different reference policies}
    \label{tab:grouping-table}
    \begin{tabular}{>{\centering\arraybackslash}p{0.12\textwidth}!{\vrule width 1pt}>{\centering\arraybackslash}p{0.26\textwidth}|>{\centering\arraybackslash}p{0.26\textwidth}|>{\centering\arraybackslash}p{0.26\textwidth}}
        \textbf{Policy} & \textbf{DoF-optimized} & \textbf{Combinatorial parallel links} & \textbf{Cyclic parallel links} \\
        \Xhline{1pt}
        \textbf{DoF} 
        & $\begin{aligned}
        &\textrm{DoF}^{*}_{\mathrm{opt},\CJ} = \\
        &\qquad \frac{K}{ \sum_{j\in\CJ}  \frac{ K_{j}}{ \Omega^*_{j} \beta^*_{j}}} 
        \end{aligned}$ \eqref{eq:total_DoF_grouping}
        & $\begin{aligned}
        &\mathrm{DoF}^{}_{\mathrm{cmb},\CJ} = \\
        &\qquad\frac{K}{\sum_{j\in\CJ}^{} \frac{K_{j}}{G_{j} K_{j} \gamma + G_{j} \big\lfloor \tfrac{L}{G_{j}} \big\rfloor}}
        \end{aligned}$ \eqref{eq:total_DoF_grouping_cpar}
        & $\begin{aligned}
        &\mathrm{DoF}^{}_{\mathrm{lin},\CJ} = \\
        &\qquad \frac{K}{\sum_{j\in\CJ}^{} \frac{K_{j}}{G_{j} K_{j} \gamma + G_{j} \big\lfloor \tfrac{L}{G_{j}} \big\rfloor}}
        \end{aligned}$\eqref{eq:total_DoF_grouping_lpar} \\ 
        \hline
        \makecell{
        \textbf{DoF} \\
        per group $j$
        }
        &$\begin{aligned}
        \textrm{DoF}^{*}_{\mathrm{opt},j} &= \Omega^*_{j} \cdot \beta^*_{j} \\
        \Omega^*_{j} , \beta^*_{j} &\leftarrow \Delta(K_{j},\gamma,L,{G}_{j})
        \end{aligned}$ 
        & $\begin{aligned}
        &\mathrm{DoF}^{}_{\mathrm{cmb},j} = \\
        &\qquad G_{j}\color{black}\big(\color{black} K_{j}\gamma + \big\lfloor \tfrac{L}{G_{j}} \big\rfloor\color{black}\big)\color{black}
        \end{aligned}$
        & $\begin{aligned}
        &\mathrm{DoF}^{}_{\mathrm{lin},j} = \\
        &\qquad G_{j} \color{black}\big(\color{black}K_{j}\gamma + \big\lfloor \tfrac{L}{G_{j}} \big\rfloor\color{black}\big)\color{black}
        \end{aligned}$
        \\
        \hline
        \makecell{
        \textbf{Subpacketization} \\
        per group $j$
        }
        & $\begin{aligned}
        &\Theta^{}_{\mathrm{opt},j} = \\
        &\quad \binom{K_{j}}{K_{j}\gamma}\binom{K_{j}-K_{j}\gamma-1}{\Omega_{j}^*-K_{j}\gamma-1}\beta_{j}^*
        \end{aligned}$
        & $\begin{aligned}
        &\Theta^{}_{\mathrm{cmb},j} = \\
        &\quad  \binom{K_{j}}{K_{j} \gamma}\binom{K_{j}-K_{j}\gamma-1}{\big\lfloor \tfrac{L}{G_{j}} \big\rfloor -1}G_{j}
        \end{aligned}$
        & $\begin{aligned}
        &\Theta^{}_{\mathrm{lin},j} = \\
        &\quad G_{j} K_{j} \left( K_{j}\gamma + \big\lfloor \tfrac{L}{G_{j}} \big\rfloor \right)
        \end{aligned}$ \\
        \hline
        \makecell{
        \textbf{Trans. count} \\
        per group $j$
        }
        & $\begin{aligned}
        &{S}_{\mathrm{opt},j} = \\
        &\qquad \binom{K_{j}}{\Omega^*_{j}} \binom{\Omega^*_{j}-1}{K_{j}\gamma}
        \end{aligned}$ 
        & $\begin{aligned}
        &S^{}_{\mathrm{cmb},j} = \\
        & \binom{K_{j}}{K_{j}\gamma+\big\lfloor \tfrac{L}{G_{j}} \big\rfloor}  \binom{K_{j}\gamma+\big\lfloor \tfrac{L}{G_{j}} \big\rfloor-1}{K_{j}\gamma}
        \end{aligned}$ 
        & $\begin{aligned}
        &S^{}_{\mathrm{lin},j} = \\
        &\qquad K_{j}\left(K_{j}-K_{j}\gamma\right)
        \end{aligned}$ \\
        \hline
        \textbf{Applicability} 
        & All network parameters 
        & All network parameters 
        & $\big\lfloor \tfrac{L}{{G}_{j}} \big\rfloor \ge K_{j}\gamma, \quad \forall j \in [J]$ \\
    \end{tabular}%
\end{table*}

\begin{lemma}\label{lm: acheivableDoFgrpop} 
The DoF of the Grouping strategy with the $\mathrm{opt}$ policy, denoted by $\mathrm{DoF}_{\mathrm{opt},\CJ}^*$, can be written as~\eqref{eq:total_DoF_grouping} in Table~\ref{tab:grouping-table}.
\begingroup
\refstepcounter{equation}\label{eq:total_DoF_grouping}
\endgroup
\end{lemma}

\begin{proof}
    With the $\mathrm{opt}$ policy, at step $j \in \CJ$ of the placement phase, each file $W \in \CF$ is split into $\vartheta_{j} = \binom{K_{j}}{K_{j}\gamma}$
    subfiles $W_{\CP_{j}}$, where $\CP_{j} \subseteq \CK_{j}$ and $|\CP_{j}| = K_{j} \gamma$, and 
    each user $k \in \CK_{j}$ stores all $W_{\CP_{j}}$ for which $k \in \CP_{j}$. Then, at step $j \in \CJ$ of the delivery phase, missing parts of the files requested by users in $\CK_j$ are delivered through ${S}_{\mathrm{opt},j}$ transmission vectors (see Table~\ref{tab:grouping-table}), where delivery parameters $\Omega^*_{j}$ and $\beta^*_{j}$ are obtained by solving $(\Omega^*_{j}, \beta^*_{j}) = \Delta(K_{j},\gamma,L,{G}_{j})$. The delivery phase requires extra subpacketization, and the final subpacketization level during step $j \in \CJ$ is given as $\Theta^{}_{\mathrm{opt},j}$ (see Table~\ref{tab:grouping-table}).
%
As a result, for each transmission $s \in [{S}_{\mathrm{opt},j}]$ in step $j \in \CJ$, we have $f(s) = \nicefrac{F}{\displaystyle \Theta^{}_{\mathrm{opt},j} }$. 
Moreover, as the transmission vectors in step $j$ deliver every missing part of all users in $\CK_{j}$, and as $\textrm{DoF}^{*}_{\mathrm{opt},j} =\Omega^*_{j}\beta^*_{j}$, we can write
\begin{equation*}\label{convk2j}
    {S}_{\mathrm{opt},j} \times \textrm{DoF}^{*}_{\mathrm{opt},j} \times \nicefrac{F}{\Theta^{}_{\mathrm{opt},j}} = K_{j}(1-\gamma)F.
\end{equation*}
Substituting these into~\eqref{eq:asym_DoF}, the DoF of the \emph{Grouping} strategy with the $\mathrm{opt}$ policy can be written as
\begin{equation*}\label{eq:total_DoF_grouping_prf}
\begin{aligned}
    \mathrm{DoF}_{\mathrm{opt},\CJ}^* &=  \frac{\sum_{j\in \CJ} S_{\mathrm{opt},j} \frac{F}{\Theta^{}_{\mathrm{opt},j}}\Omega^*_{j}\beta^*_{j}
    }{ \sum_{j\in\CJ} S_{\mathrm{opt},j} \frac{F}{\Theta^{}_{\mathrm{opt},j}}
    }
    =  
    \frac{K
    }{ \sum_{j\in\CJ}  \frac{ K_{j}}{ \Omega^*_{j} \beta^*_{j}}
    }, 
\end{aligned}
\end{equation*}
and the proof is complete. 
\end{proof}

\begin{lemma}\label{lm: acheivableDoFgrpop_cpar} 
The DoF of the Grouping strategy with both $\mathrm{cmb}$ and $\mathrm{lin}$ policies, denoted by $\mathrm{DoF}^{}_{\pi,\CJ}$, $\pi \in \{\mathrm{cmb},\mathrm{lin}\}$, is given by~\eqref{eq:total_DoF_grouping_cpar} and~\eqref{eq:total_DoF_grouping_lpar} in Table~\ref{tab:grouping-table}. 
\begingroup
\refstepcounter{equation}\label{eq:total_DoF_grouping_cpar}
\endgroup
\begingroup
\refstepcounter{equation}\label{eq:total_DoF_grouping_lpar}
\endgroup
\end{lemma}

\begin{proof}
The proof follows a similar structure as the proof of Lemma~\ref{lm: acheivableDoFgrpop}. Noting that the final subpacketization level at step $j \in \CJ$ with policy $\pi$ is $\Theta_{\pi,j}$, for each transmission $s \in [S_{\pi,j}]$ in step $j$ we have $f(s)\!=\!\nicefrac{F}{\Theta_{\pi,j}}$. Moreover, as 
%
%
 the transmission vectors in step $j$ deliver all the missing parts of all users in $\CK_{j}$ with $\textrm{DoF}^{}_{\pi,j}\!=\!G_{j} (K_{j}\gamma + \lfloor\! \nicefrac{L}{G_{j}}\! \rfloor)$, we can write
\begin{equation*}\label{convk3j}
\begin{aligned}
{S}_{\pi,j} \textrm{DoF}^{}_{\pi,j} \times \nicefrac{F}{\Theta_{\pi,j}} = K_{j}(1-\gamma)F.
   \end{aligned}
\end{equation*}
Substituting this into~\eqref{eq:asym_DoF}, the DoF of the \emph{Grouping} strategy with policy $\pi \in \{\mathrm{cmb},\mathrm{lin}\}$
can be expressed as
\begin{equation}\label{eq:total_DoF_grouping_cpar1}
\begin{aligned}
\mathrm{DoF}_{\pi,\CJ} \!&=\!  \frac{\sum_{j\in \CJ} S_{\pi,j} \frac{F}{\Theta_{\pi,j}}G_{j} (K_{j}\gamma + \lfloor \nicefrac{L}{G_{j}} \rfloor)
    }{ \sum_{j\in\CJ} S_{\pi,j} \frac{F}{\Theta_{\pi,j}}
    },
\end{aligned}
\end{equation}
which, after simplification, results in~\eqref{eq:total_DoF_grouping_cpar},
and the proof is complete. 
\end{proof}
\color{black}Overall, 
neither \emph{min-$G$} nor \emph{Grouping} universally dominates. As quantified by Tables~\ref{tab:minG-table} and~\ref{tab:grouping-table}, their relative achievable DoF depends on the interplay among the group sizes, receive dimensions, and cache gain.\color{black}
\color{black}
\section{Hybrid asymmetric solutions}
In this section, we study two composite strategies that enable an interplay between CC and spatial multiplexing gains in asymmetric MIMO-CC systems.
\subsection{The Super-grouping strategy}
\label{super_grp_sub}
The \emph{Grouping} strategy, as introduced in Section~\ref{subsec: grouping}, treats each set of users $\CK_j$ separately. A key limitation of this approach 
is the potential degradation in the overall achievable DoF due to relying solely on maximizing spatial multiplexing gain at the Rx side. To overcome this limitation, we propose a promising alternative, known as the \emph{Super-grouping} strategy, that
combines the \emph{min-$G$} and \emph{Grouping} strategies into a unified framework. The first step in the \emph{Super-grouping} strategy is to merge the existing $J$ user sets $\CK_j$ into $\bar{J}\leq J$ equivalent user sets $\bar{\CK}_{l}$, $l\in[\bar J]$. Then, we employ the \emph{min-$G$} strategy to deliver data to users inside each equivalent user set, and use the \emph{Grouping} strategy among different sets. 

\vspace{1pt}
\noindent\textbf{Construction of equivalent user sets.} 
As the \emph{Super-grouping} strategy employs the \emph{min-$G$} strategy to deliver data to users within each equivalent user set $\bar{\CK}_{l}$, from the DoF perspective, it is preferable to merge original sets with consecutive indices only. This is because the \emph{min-$G$} strategy sets the Rx spatial gain to the smallest Rx gain among the involved users; hence, with the original sets $\CK_j$ ordered by the users' Rx gains, merging non-consecutive indices inevitably translates into a larger DoF loss than consecutive-index merging.

For consecutive-index merging, we need to first find $\bar{J}-1$ indices $J_1,\cdots,J_{\bar{J}-1}$ such that $1 \le J_1 < J_2 < \cdots <J_{\bar{J}-1} < J$. Then, defining $J_0 = 0$ and $J_{\bar{J}} = J$, we can create $\bar{J}$ index sets as
\begin{equation}
\begin{array}{l}\label{sgrp_delivery1}
\CJ_l = \{J_{l-1}+1, \ldots, J_{l}\}, \qquad l\in [\bar{J}]  
\end{array}
\end{equation}
and create the equivalent user sets as
\begin{equation}
    \bar{\CK}_{l} = \bigcup_{j \in \CJ_{l}} \CK_{j}.
\end{equation}
For a given $\bar{J} \in [J]$, the number of ways to select $J_1,\cdots,J_{\bar{J}-1}$ such that $1 \le J_1 < J_2 < \cdots <J_{\bar{J}-1} < J$ is $\binom{J-1}{\bar{J} - 1}$. 
Each selection of these indices defines a unique merging of the original user sets into equivalent sets. In order to find the 
\color{black}{largest achievable DoF} \color{black}
with the \emph{Super-grouping} strategy, one should compare all these merging possibilities and select the one with the largest DoF.
It should be noted that, when $\bar{J} = 1$ and $\bar{J} = J$, the \emph{Super-grouping} strategy reduces to \emph{min-$G$} and \emph{Grouping}, respectively.
\begin{table*}[t]
    \centering
    \caption{The \emph{Super-grouping} strategy parameters for different reference policies}
    \label{tab:super-grouping-table}
    \begin{tabular}{>{\centering\arraybackslash}p{0.12\textwidth}!{\vrule width 1pt}>{\centering\arraybackslash}p{0.26\textwidth}|>{\centering\arraybackslash}p{0.26\textwidth}|>{\centering\arraybackslash}p{0.26\textwidth}}
        \textbf{Policy} & \textbf{DoF-optimized} & \textbf{Combinatorial parallel links} & \textbf{Cyclic parallel links} \\
        \Xhline{1pt}
        \textbf{DoF} 
        & $\begin{aligned}\\[-7mm]
        &\textrm{DoF}^{*}_{\mathrm{opt},\bar{\CJ}^*} = \\
        &\qquad \frac{K
    }{ \sum_{l\in[\bar{J}^*]} \displaystyle \frac{ \sum_{j\in \CJ_{l}^*} K_{j}}{ \bar \Omega^*_{{l}}\bar \beta^*_{{l}}}
    } \end{aligned}$ 
        & $\begin{aligned}\\[-3.8mm]
        &\mathrm{DoF}^{}_{\mathrm{cmb},\bar{\CJ}^*} = \\
        &\hspace{-2mm}\frac{K}{\sum\limits_{l\in[\bar{J}^*]}^{} \frac{\!\!\!\sum\limits_{j\in \CJ_{l}^*}\!\!\!K_{j}}{\min\limits_{j\in \CJ_{l}^*}\! G_{j}\!\big(\!\!\sum\limits_{j\in \CJ_{l}^*}\!\!\!K_{j} \gamma + \big\lfloor\! \tfrac{L}{\min\limits_{j\in \CJ_{l}^*}\!\! G_{j}\!} \!\big\rfloor\big)}}
        \end{aligned}$ 
        & $\begin{aligned}\\[-3.8mm]
        &\mathrm{DoF}^{}_{\mathrm{lin},\bar{\CJ}^*} = \\
        &\hspace{-2mm}\frac{K}{\sum\limits_{l\in[\bar{J}^*]}^{} \frac{\!\!\!\sum\limits_{j\in \CJ_{l}^*}\!\!\!K_{j}}{\min\limits_{j\in \CJ_{l}^*}\! G_{j}\!\big(\!\!\sum\limits_{j\in \CJ_{l}^*}\!\!\!K_{j} \gamma + \big\lfloor\! \tfrac{L}{\min\limits_{j\in \CJ_{l}^*}\!\! G_{j}\!} \!\big\rfloor\big)}}
        \end{aligned}$ \\ 
        \hline
        \textbf{Applicability} 
        & All network parameters 
        & All network parameters 
        & $\big\lfloor \tfrac{L}{\min\limits_{j\in \CJ_{l}^*}\! G_{j}} \big\rfloor \! \ge \!\!\sum\limits_{j\in \CJ_{l}^*}\!\!K_{j}
\gamma, \quad \forall l \in [\bar{J}^*]$ \\
    \end{tabular}%
\end{table*}
Based on this core concept, the analytical results for the proposed \emph{Super-grouping} strategy are developed and summarized in Table~\ref{tab:super-grouping-table}.

\vspace{1pt}
\noindent\textbf{Placement phase.} 
Data placement is done in $\bar{J}$ steps, where at each step $l\in[\bar{J}]$, cache memories of users in equivalent set $\bar \CK_{{l}}$ are filled with data, according to the selected reference policy. It should be noted that, in step $l$, the $K$ value in the reference policy should be replaced with
\begin{equation}
\begin{array}{l}
\bar K_{{l}} = \sum\limits_{j\in \CJ_{l}} K_{j}.
\end{array}
\end{equation}

\vspace{1pt}
\noindent\textbf{Delivery phase.} Data delivery proceeds in $\bar{J}$ steps. In step $l\in[\bar{J}]$, all missing parts of the data files requested by users in $\bar{\CK}_l$ are delivered, following the selected reference policy and by replacing $G$ with
\begin{equation}
\begin{array}{l}
\bar G_{{l}} = \min\limits_{j\in \CJ_{l}} G_{j}.
\end{array}
\end{equation} 

\vspace{1pt}
 \noindent\textbf{DoF analysis.} The maximum achievable 
%
DoF of the \emph{Super-grouping} strategy for a given policy $\pi$, denoted by
$\mathrm{DoF}_{\pi,\bar{\CJ}^*}$,
is obtained by jointly optimizing the number of equivalent user sets $\bar J$ and the way the user sets $\{\CK_j\}_j$ are merged to form the equivalent sets:
\begin{equation}\label{grp_delivery34_opt_concise}
(\bar{J}^*,\mathbf{\bar{J}}^*)=
\arg\max_{\substack{\\[.5ex]\bar{J}\in[J]}}\;\;\;
\max_{\substack{\\[.5ex]J_{l-1} < J_l < J_{l+1},\\[.5ex] \forall l\in[\bar J-1]}}
\mathrm{DoF}_{\pi,\bar{\CJ}},
\end{equation}
where $\bar \BJ = [0,J_1,\cdots,J_{\bar{J}-1},J]$ includes the partitioning indices used for the merging process. In~\eqref{grp_delivery34_opt_concise}, the
inner maximization searches over the $\bar J-1$ splitting points for the chosen $\bar J$, and outer maximization aims to find the best number of partitions $\Bar J$. Therefore, we obtain the optimized Super-group sets as $\CJ_{l}^*$.
\subsection{The Phantom strategy}
\label{subsec: phantom}
The \emph{Phantom} strategy also serves as a bridge between the \emph{min-$G$} and \emph{Grouping} strategies. This method maximizes the CC gain while leveraging the excess spatial multiplexing of users with more antennas than the minimum value, which would otherwise remain unused in the \emph{min-$G$} strategy. \color{black}The Phantom construction is a symmetrization-and-projection mechanism inspired by virtual-user techniques in coded caching~\cite{salehi2020lowcomplexity,abolpour2024resource}. \color{black}
%

The \emph{Phantom} strategy involves multiple multicasting (MC) rounds, followed by a single unicasting (UC) round. An MC round $i$ includes $\hat{S}_{\pi,i}^{\mathrm{MC}}$ transmission intervals, where $\pi$ is the selected reference symmetric transmission policy (see Section~\ref{ref-schemes}). MC transmissions in round $i$ deliver parts of the missing data to a subset of users $\hat{\CK}_i$. In the first round, $\hat{\CK}_1=\CK$ involves all the users in the network, and in each subsequent round, the set of selected users gets smaller. The MC rounds continue until no users are left for MC transmissions. We use $I$ to denote the total number of MC rounds. 

As a general overview, in round $i \in [I]$, data delivery for users in $\hat{\CK}_i$ proceeds as follows. First, assuming a hypothetical symmetric setup similar to the original network, but with $\hat{G}_i = \max_{k \in \hat{\CK}_i} \SfG_k$ Rx antennas at each user $k \in \hat{\CK}_i$, a delivery algorithm is run according to the selected reference policy $\pi$. This results in a total number of ${S}_{\pi,i}^{\mathrm{MC}}$ candidate transmission vectors $\tilde{\Bx}_{\pi,i}^{\mathrm{MC}}(s)$, where $s \in [{S}_{\pi,i}^{\mathrm{MC}}]$ shows the interval index (note that $s$ is re-initialized at the beginning of each round). Each candidate transmission vector is capable of delivering $\hat{\beta}_{\pi,i}$ parallel streams to a subset of users, with size $\hat{\Omega}_{\pi,i}$, in the hypothetical network.\footnote{
%
%
The ``candidate transmission vectors'' are not transmitted directly; in the next step, a subset of subpackets is removed to ensure linear decodability.
%
}
Next, each candidate vector $\tilde{\Bx}_{\pi,i}^{\mathrm{MC}}(s)$ is mapped to a real transmission vector $\hat{\Bx}_{\pi,i}^{\mathrm{MC}}(s)$. Since some users in the original network may have fewer than $\hat{\beta}_{\pi,i}$ Rx antennas, a subset of their intended subpackets are removed from the candidate transmission vectors to ensure linear decodability for all $k \in \hat{\CK}_i$. The 
subpackets \color{black}removed \color{black} \color{black} from the current candidate transmission vector \color{black} are deferred to later transmission, either in a subsequent MC round or in the final UC round. The term Phantom stems from the assumption of symmetric number of antennas across the selected users in each round: users with fewer antennas are effectively treated as if they possessed additional, non-existent (phantom) receive antennas, whose associated streams are subsequently
\color{black}removed from the candidate
transmission vector\color{black}.
In particular, the excess components $\!(\hat\beta_{\pi,i}\!-\! \SfG_k\!)^{+}$ correspond to non-existing receive dimensions and are removed from the current candidate
transmission vector, rather than discarded, and are deferred to later MC rounds or to the UC phase. Lemma~\ref{LD_Phantom} proves that this
projection preserves linear decodability.
The proposed \emph{Phantom} strategy is 
detailed in the following.

%
\vspace{1pt}
\noindent\textbf{Placement phase.} Follows the reference policy $\pi$, as detailed in Section~\ref{ref-schemes}. We use $W_{\wp_\pi}$ to show a packet of a general file $W$, where $\wp_\pi$ clarifies which users store $W_{\wp_\pi}$ in their cache memories. For $\pi \in \{\mathrm{opt},\mathrm{cmb}\}$, $\wp_\pi$ is a subset of user indices, while for $\pi = \mathrm{lin}$, it represents a single user index.

\vspace{1pt}
\noindent\textbf{Delivery phase, MC rounds.} At each MC round $i \in [I]$, the delivery algorithm proceeds as follows:

\vspace{1pt}
\noindent\textbf{Step~1 - Target user set determination.} For $i = 1$, the target user set includes all users, i.e., $\hat{\CK}_1 = \cup_{j \in [J]} \CK_j=\CK$. 
For $i > 1$, we have:
\begin{equation}
    \hat{\CK}_i = \bigcup\nolimits_{j \in [J], \; G_j < \hat{\beta}_{\pi,i-1}} \CK_j.
\end{equation}
The rationale for this selection is that after round~$i-1$, a subset of missing subpackets of users with fewer Rx antennas than $\hat{\beta}_{\pi,i-1}$ are \color{black}deferred\color{black}, and in round~$i$, we aim to find multicasting opportunities among such \color{black}deferred \color{black} subpackets.

\vspace{1pt}
\noindent\textbf{Step~2 - Candidate transmission vector creation.} 
\color{black}Retaining the original coded-caching gain $K\gamma$, \color{black} we consider a hypothetical symmetric MIMO-CC setup, with $L$ antennas at the transmitter and $\hat K_i = |\hat{\CK}_i|$ users, each with cache ratio of $\gamma$ and $\hat{G}_i = \max_{k \in \hat{\CK}_i} \SfG_k$ Rx antennas. For this hypothetical setup, we follow the reference policy $\pi$, to create a total number of ${S}_{\pi,i}^{\mathrm{MC}}$ candidate transmission vectors. First, following the discussions in Section~\ref{ref-schemes}, we find out the number of users served in each transmission, denoted by $\hat{\Omega}_{\pi,i} \leq \hat K_i$, and the number of parallel streams per user, denoted by $\hat{\beta}_{\pi,i} \le \hat{G}_i$\footnote{\color{black}
Reducing parallel streams to extend 
the
feasibility
region
(as discussed in Remark~\ref{Flexible-stream cyclic operation})
also applies to each \emph{Phantom} MC round,
replacing $\hat G_i$ with feasible $\hat\beta_{\rm lin,i}$
in the $\rm lin$ expressions.
%
%
%
}. 
%
%
Then, subpacketization is increased by a factor of $\hat{\phi}_{\pi,i}$, as:
\begin{equation}\label{eq:policy_phi_phantom}
\begin{aligned}
\hat{\phi}_{\mathrm{opt},i}
&= \hat{\beta}_{\mathrm{opt},i}^{*}
   \binom{\hat{K}_i - K \gamma - 1}
         {\hat{\Omega}_{\mathrm{opt},i}^{*} - K\gamma - 1}, \\
\hat{\phi}_{\mathrm{cmb},i}
&= \hat{G}_i
    \binom{\hat{K}_i - K \gamma - 1}
         {\left\lfloor \frac{L}{\hat{G}_i} \right\rfloor - 1},
         \\
\hat{\phi}_{\mathrm{lin},i}
&= \hat{G}_i
   \Bigl( K \gamma + \big\lfloor  \frac{L}{\hat{G}_i}  \big\rfloor \Bigr).
%
\end{aligned}
\end{equation}\color{black}
It is important to note that, as in each round $i$ we aim to find MC opportunities among subpackets \color{black}deferred \color{black} in the previous round $i-1$, the subpacketization accumulates multiplicatively across rounds. Hence, the final subpacketization value for data sent in round $i$ is given as
\begin{equation}
    \hat{\Theta}_{\pi,i} = {\vartheta}_\pi \prod_{i' \in [i]} \hat{\phi}_{\pi,i'},
\end{equation}
where ${\vartheta}_\pi$ denotes the subpacketization of the placement phase.
This accumulated subpacketization must be accounted for while calculating the number of intervals in round $i$, i.e., $S_{\pi,i}^{\mathrm{MC}}$.  Specifically, in the first MC round, $S_{\pi,1}^{\mathrm{MC}} = S_{\pi,1}$, and for $i > 1$ we have
\begin{equation}\label{eq:MC_round_trans_count}
    S_{\pi,i}^{\mathrm{MC}} = \Bigg(\prod_{i' \in [i-1]} \hat{\phi}_{\pi,i'}\Bigg) S_{\pi,i},
\end{equation}
where $S_{\pi,i}$ is the number of transmissions for the hypothetical network in round $i$, using the reference policy $\pi$, as defined in Section~\ref{ref-schemes}. For each interval $s \in [S_{\pi,i}^{\mathrm{MC}}]$, the candidate vector can be expressed as
\begin{equation}
    \tilde{\Bx}_{\pi,i}^{\mathrm{MC}}(s) = \sum\nolimits_{k \in \hat{\CK}_{\pi,i}(s)} \sum\nolimits_{W_{\wp_\pi}^q \in \tilde{\CM}_{\pi,i}^k(s)}   \Bw_{\wp_\pi}^q W_{\wp_\pi}^q,
\end{equation}
where $\hat{\CK}_{\pi,i}(s) \subseteq \hat{\CK}_i$ denotes the set of users receiving data in transmission interval $s$ of round $i$, $\tilde{\CM}_{\pi,i}^k(s)$ represents the set of missing subpackets of user $k$ transmitted in this interval (it includes $\hat{\beta}_{\pi,i}$ subpackets), $q$ is an auxiliary index to account for delivery-phase subpacketization, and $\Bw_{\wp_\pi}^q$ is the transmit precoder 
vector for subpacket $W_{\wp_\pi}^q \in \tilde{\CM}_{\pi,i}^k(s)$.\footnote{The choice of the precoder design (e.g., zero-forcing or MMSE) affects finite-SNR performance, but does not change the DoF. Interference suppression criteria are given by the selected reference policy $\pi$.} 


\vspace{1pt}
\noindent\textbf{Step~3 - Final transmission vector creation.} Each candidate vector $\tilde{\Bx}_{\pi,i}^{\mathrm{MC}}(s)$ created in Step~2 aims to deliver $\hat{\beta}_{\pi,i}$ parallel streams to each user $k \in \hat{\CK}_{\pi,i}(s)$. However, due to asymmetric Rx capabilities, there might exist some users in $\hat{\CK}_{\pi,i}(s)$ with fewer than $\hat{\beta}_{\pi,i}$ Rx antennas. In such cases, to ensure linear decodability by all users, for each user $k \in \hat{\CK}_{\pi,i}(s)$ with $\SfG_k < \hat{\beta}_{\pi,i}$, we randomly \color{black}remove \color{black} $\hat{\beta}_{\pi,i} - \SfG_k$ subpackets from $\tilde{\CM}_{\pi,i}^k(s)$, and leave these subpackets to be transmitted later, either in subsequent MC rounds or in the final UC round. Denoting the reduced subpacket set by $\hat{\CM}_{\pi,i}^k(s)$, we can write down the real transmission vectors as
\begin{equation}\label{eq:real_trans_vec}
    \hat{\Bx}_{\pi,i}^{\mathrm{MC}}(s) = \sum\nolimits_{k \in \hat{\CK}_{\pi,i}(s)} \sum\nolimits_{W_{\wp_\pi}^q \in \hat{\CM}_{\pi,i}^k(s)}   \Bw_{\wp_\pi}^q W_{\wp_\pi}^q
\end{equation}

\color{black}

\vspace{1pt}
\noindent\textbf{Delivery phase, UC round.}
After completing all MC rounds, the remaining subpackets are delivered through $S_{\pi}^{\textrm{UC}}$ UC transmission intervals. This round does not introduce extra subpacketization, i.e., $\hat{\Theta}_{\pi, I}$ is used for UC transmissions as well. 
\color{black} Consistent with the high-SNR DoF setting of Remark~\ref{standard_dof_analysis}, a scheduled user set $\Bar{\CK}$ can support at most $\min (L, \sum_{k \in \Bar{\CK}} \SfG_{k})$ parallel streams. 
Thus, 
standard linear ZF precoding
can deliver up to \(L\) residual subpackets in parallel subject to the per-user receive-dimension constraints. 

Let $\mathcal R_k$ denote the set of residual subpackets required by user $k$, and define \[ \CK_{\rm UC} \triangleq \{k\in\CK:\mathcal R_k\neq\emptyset\}, \quad d_k\triangleq|\mathcal R_k|, \quad Z_{\pi}^{\rm UC} \triangleq \sum_{k\in\CK_{\rm UC}}d_k. \] 
The UC completion time is jointly constrained by the BS transmit dimension and the users' receive dimensions and is given by
\begin{equation}\label{RA_UCT} S_{\pi}^{\rm UC} = \max\Bigg\{ \underbrace{\Big\lceil \nicefrac{Z_{\pi}^{\rm UC}}{L} \Big\rceil}_{S_{\pi,L}^{\rm UC}}, \; \max_{k\in\CK_{\rm UC}} \Big\lceil \nicefrac{d_k}{\SfG_k} \Big\rceil \Bigg\}. \end{equation}
The first term $S_{\pi,L}^{\rm UC}$ is imposed by the BS-side constraint that at most $L$ residual streams can be transmitted per UC interval, whereas the second follows from the fact that user $k$ can receive at most $\SfG_k$ residual streams per interval. 

To characterize when the user-side can be ignored, let $\CJ_{\rm UC}$ denote the groups requiring UC transmission, with common residual demand $d_j$ and receive dimension $G_j$ within group $j$, such that $Z_{\pi}^{\rm UC}=\sum_{j\in\CJ_{\rm UC}}K_jd_j$.
For a limiting group $ j^\star\in \arg\max_{j\in\CJ_{\rm UC}}\frac{d_j}{G_j} $, the BS-side determines the UC completion time if and only if 
\begin{equation}\label{eq:uc_rxa_activation} 
\Big\lceil \nicefrac{d_{j^\star}}{G_{j^\star}} \Big\rceil \leq \Big\lceil \nicefrac{\big(K_{j^\star}d_{j^\star}+Z_{-j^\star}^{\rm UC}\big)}{L} \Big\rceil, 
\end{equation} 
where $Z_{-j^\star}^{\rm UC} \triangleq \sum_{\substack{j\in\CJ_{\rm UC}\\j\neq j^\star}}K_jd_j.$ 
Accordingly, the user-side term can become active only if
%
\begin{equation}\label{eq:uc_rxa_structure}
\begin{array}{l}
\frac{d_{j^\star}}{G_{j^\star}}
>
\frac{K_{j^\star}d_{j^\star}+Z_{-j^\star}^{\rm UC}}{L}
\Longleftrightarrow
\left(
L-K_{j^\star}G_{j^\star}
\right)d_{j^\star}
>
G_{j^\star}Z_{-j^\star}^{\rm UC}.
\end{array}
\end{equation}
which requires
\begin{equation}\label{necc_1}
    K_{j^\star}G_{j^\star}<L
\end{equation}
and 
\begin{equation}\label{necc_2}
\begin{array}{l}
d_{j^\star}
>
\frac{G_{j^\star}}
{L-K_{j^\star}G_{j^\star}}
Z_{-j^\star}^{\rm UC}.
\end{array}
\end{equation}
\color{black} When~\eqref{eq:uc_rxa_structure} is inactive,
the UC completion time reduces to
$S_{\pi}^{\rm UC} = S_{\pi,L}^{\rm UC} = \big\lceil\nicefrac{Z_{\pi}^{\rm UC}}{L}\big\rceil$;
otherwise, the general expression
in~\eqref{RA_UCT} applies. \color{black}
%
 %
\begin{lemma}\label{LD_Phantom}
    For all $\pi \in \{\mathrm{opt},\mathrm{cmb},\mathrm{lin}\}$, $i \in [I]$, and $s \in [S_{\pi,i}^{\mathrm{MC}}]$, each transmission vector $\hat{\Bx}_{\pi,i}^{\mathrm{MC}}(s)$ in~\eqref{eq:real_trans_vec} is linearly decodable by all users $k \in \hat{\CK}_{\pi,i}(s)$.
\end{lemma}
\begin{proof}
     The proof is given in Appendix~\ref{app_B}. 
\end{proof}
\noindent\textbf{DoF analysis.} In order to calculate the DoF of the \emph{Phantom} strategy, we start by the DoF expression for non-symmetric setups in~\eqref{eq:asym_DoF_phantom_version}. The main challenge is to determine the total number of time intervals; while the number of intervals for each MC round is given by~\eqref{eq:MC_round_trans_count}, for the number of UC intervals requires first calculating the total number of subpackets remaining for the UC round, denoted by $Z_{\pi}^{\mathrm{UC}}$. This can be done by summing up the number of subpackets delivered within all MC rounds, and substituting the result from the total number of missing subpackets. Since the subpacketization increases multiplicatively across consecutive MC rounds, all quantities must be expressed with respect to a common reference. To this end, we
introduce normalization factors $\epsilon_{\pi,i} \geq 1$ that normalize the subpacketization in each MC round to the maximum (final) degree of subpacketization, such that $\epsilon_{\pi,i}\hat{\Theta}_{\pi,i} = \hat{\Theta}_{\pi,I}$, $\forall i \in [I]$,
%
which simply results in $\epsilon_{\pi,i} = \prod_{i < i' \le I} \hat{\phi}_{\pi,i'}$. Indeed, this also requires normalizing the number of transmission intervals in each MC round, i.e., we use $\hat{S}_{\pi,i}^{\mathrm{MC}} = \epsilon_{\pi,i} {S}_{\pi,i}^{\mathrm{MC}}$ as the normalized transmission count when working with normalized subpacketization. For notational simplicity, we also define $\zeta_{\pi,i} = \nicefrac{\prod_{i' \in [I]} \hat{\phi}_{\pi,i'}}{\hat{\phi}_{\pi,i}}$, which, together with~\eqref{eq:MC_round_trans_count}, enables us to calculate the normalized transmission count of MC rounds from the transmission count of the reference policy as
\begin{equation}
    \hat{S}_{\pi,i}^{\mathrm{MC}} = \epsilon_{\pi,i} {S}_{\pi,i}^{\mathrm{MC}} = \zeta_{\pi,i} S_{\pi,i}.
\end{equation}
\begin{theorem}\label{thm:phantom_dof}
 For a given reference policy $\pi$, \color{black} the achievable
DoF of the \emph{Phantom} strategy is given by
\begin{equation}\label{prmdof_ph}
\mathrm{DoF}_{\pi,\mathrm{Ph}}
=
\frac{Z_{\pi}^{\mathrm{Tot}}}
{\displaystyle
\sum\nolimits_{i\in[I]}\zeta_{\pi,i}S_{\pi,i}
+S_{\pi}^{\mathrm{UC}}},
\end{equation}
 where $S_{\pi}^{\mathrm{UC}}$ is given
by~\eqref{RA_UCT}. \color{black} Under~\eqref{eq:uc_rxa_activation} and neglecting the UC rounding overhead, \color{black}the DoF of Phantom can be approximated as 
%
\color{black}
\color{black}
\begin{equation}\label{eq:DoF_phantom_general}
\mathrm{DoF}_{\pi,\mathrm{Ph}}
\!\approx\! \frac{Z_{\pi}^{\mathrm{Tot}}}
{\!\!\displaystyle\sum_{i\in[I]} \!\!\zeta_{\pi,i} S_{\pi,i}
+ \!\frac{1}{L}\!\! \left(\!\!
Z_{\pi}^{\mathrm{Tot}}
- \sum_{i\in[I]} \zeta_{\pi,i} \lambda_{\pi,i}
\sum_{k\in \hat{\CK}_{i}} \bm{\beta}_{\pi,i}^{k}
\!\!\right)},
\end{equation}\color{black}
    where $Z^{\mathrm{Tot}}_{\pi} = K(1-\gamma)\hat{\Theta}_{\pi,I}$ is the (normalized) total number of missing subpackets, and ${\bm\beta}_{\pi,i}^k = \min(\hat{\beta}_{\pi,i},\SfG_k)$. Moreover, $\lambda_{\pi,i}$ denotes the total number of transmissions in MC round $i$ that serve a user $k \in \hat{\CK}_i$ with reference policy $\pi$, as defined in Section~\ref{ref-schemes}.
\end{theorem}
\begin{proof}
    Starting with~\eqref{eq:asym_DoF_phantom_version}, we can see that with uniform (normalized) subpacketization at all intervals, $f(s)$ becomes independent of the interval index $s$, and the general DoF formula in~\eqref{eq:asym_DoF_phantom_version} is simplified as 
    \begin{equation}
        \mathrm{DoF} = \frac{\sum_{s \in [S]} \sum_{k \in \CK(s)}\bm\beta_k(s)}{\sum_{s \in [S]} 1},
    \end{equation}
    which can be interpreted as the (normalized) total number of missing subpackets, i.e., $Z^{\mathrm{Tot}}_{\pi}$, divided by the total number of intervals. 
    To calculate the total number of intervals, we know that
    the number of MC intervals is given as $\sum_{i \in [I]} \hat{S}_{\pi,i}^{\mathrm{MC}} = \sum_{i \in [I]} \zeta_{\pi,i}S_{\pi,i}$  \color{black} together with $S_{\pi}^{\mathrm{UC}}$ from ~\eqref{RA_UCT}, proving~\eqref{prmdof_ph}. \color{black} 
    Under~\eqref{eq:uc_rxa_activation}, UC completion time reduces to
\(S_{\pi}^{\rm UC}\!=\!S_{\pi,L}^{\rm UC}
\!\!=\!\!\lceil Z_{\pi}^{\rm UC}/L\rceil\). 
\color{black}
%
The ceiling incurs less than one interval
of rounding overhead, with a relative DoF
deviation below $1/S_{\pi}$, where $S_{\pi}$
is the exact total number of delivery
intervals. Thus, its impact is negligible
when $S_{\pi}$ is sufficiently large;
hence, the normalized UC contribution is approximated by \color{black} 
\(Z_{\pi}^{\rm UC}/L\).
    \color{black} We first note that
    \begin{equation}\label{eq:phantom_UC_packet_balance}
        Z_{\pi}^{\mathrm{UC}} = Z^{\mathrm{Tot}}_{\pi} - \sum\nolimits_{i \in [I]} \hat{Z}_{\pi,i}^{\mathrm{MC}},
    \end{equation}
    where
    \begin{equation}\label{eq:phantom_MC_packet_count}
        \hat{Z}_{\pi,i}^{\mathrm{MC}} = \zeta_{\pi,i} \lambda_{\pi,i}\sum\nolimits_{k\in\hat{\CK_i}}\bm\beta_{\pi,i}^k
    \end{equation}
    is the (normalized) total number of subpackets delivered in MC round $i$.
    Plugging these in the DoF formula results in~\eqref{eq:DoF_phantom_general}, completing the proof.
\end{proof}

\begin{remarknum}\label{phantom_polciies}
By substituting policy-specific parameters in~\eqref{eq:DoF_phantom_general}, one can 
obtain the corresponding
expression of the DoF with the \emph{Phantom} strategy. For the $\mathrm{opt}$ policy, using $\hat{\phi}_{\mathrm{opt},i}$ from~\eqref{eq:policy_phi_phantom} and
\begin{equation*}
    \begin{array}{l}
        S_{\mathrm{opt},i} = \binom{\hat{K}_i}{\hat{\Omega}_{\rm opt,i}^{*}} \binom{\hat{\Omega}_{\rm opt,i}^*-1}{K\gamma}, \; \lambda_{\mathrm{opt},i} = \binom{\hat{K}_{i}-1}{\hat{\Omega}^*_{\rm opt,i}-1} \binom{\hat{\Omega}^*_{\rm opt,i}-1}{K\gamma},
    \end{array}
\end{equation*}
the DoF can be simplified as (derivations are removed due to lack of space)
\begin{equation}\label{Imp_Phantom_DoF}
\mathrm{DoF}_{\mathrm{opt},\mathrm{Ph}}
=\frac{L}{1-\frac{1}{K} \sum\limits_{i\in [I]}^{}\!\frac{f(\hat{K}_{i}, K)}{\hat{\Omega}^*_{\rm opt,i}\hat{\beta}^*_{\rm opt,i}}\Big(\!\hat{\Omega}^*_{\rm opt,i}\!\!\sum\limits_{k\in \hat{\CK}_{i}}\!\!\bm\beta_{\mathrm{opt,i}}^k - L\hat{K}_{i}\Big)},
\end{equation}
where
\begin{equation}\label{multiplicative_ratio}
    f(\hat{K}_{i}, K)  = \frac{\binom{ \hat{K}_{i}-1}{K\gamma}}{\binom{K - 1}{K\gamma}} = \prod_{l = 0}^{K\gamma - 1} \frac{\hat{K}_{i} - 1 - l}{
    K - 1 - l}.
\end{equation}
Similarly, for the $\mathrm{cmb}$ policy, the DoF expression can be found by replacing $\hat{\Omega}_{\rm opt,i}^*$ and $\hat{\beta}_{\rm opt,i}^*$ in~\eqref{Imp_Phantom_DoF} with $K\gamma \!+\! \lfloor \nicefrac{L}{\hat{G}_i} \rfloor$ and $\hat{G}_i$, respectively. On the other hand, 
for the $\mathrm{lin}$ policy, we can use
\begin{equation*}
    \begin{aligned}
        S_{\mathrm{lin},i} = \hat{K}_i (\hat{K}_i-K\gamma), \; \lambda_{\mathrm{lin},i} = (\hat{K}_i-K\gamma)(K\gamma + \lfloor \nicefrac{L}{\hat{G}_i} \rfloor),
    \end{aligned}
\end{equation*}
to achieve
\begin{equation}\label{Imp_Phantom_DoFlin}
\hspace{-2mm}\mathrm{DoF}_{\mathrm{lin},\mathrm{Ph}}=\frac{L}{1-
\!\!\sum\limits_{i\in[I]}^{}\!\frac{(\hat{K}_{i} -K\gamma)}{K^2(1-\gamma)\hat{G}_{i}}\Big(\!\sum\limits_{k\in \hat{\CK}_{i}}\!\!\bm\beta_{\mathrm{lin,i}}^k-\frac{L\hat{K}_{i}}{ \big\lfloor \frac{L}{\hat{G}_{i}} \big\rfloor + K\gamma}\Big)}\,.
\end{equation}
\end{remarknum}
\begin{remarknum}\label{s-r-ph}
The MC phase may be terminated after any $\hat{I}<I$ rounds, with the remaining subpackets delivered via UC, at the cost of reduced DoF. However, from~\eqref{multiplicative_ratio}, for $\pi \in \{\mathrm{opt},\mathrm{cmb}\}$, increasing $i$ reduces $f(\hat{K}_i,K)$ (since $\hat{K}_i>\hat{K}_{i+1}$), with a stronger impact for larger $K\gamma$. As a result, for sufficiently large $K\gamma$, the DoF of the $\mathrm{opt}$ and $\mathrm{cmb}$ policies may already slightly decrease at $\hat{I}=1$, as is also observed numerically in the next section. The same trend applies to $\mathrm{lin}$, albeit with a larger DoF loss due to the absence of a multiplicative factor.
Motivated by its reduced complexity and subpacketization, we focus on the single-round ($\hat{I}=1$) Phantom strategy ($\mathrm{SPh}$). Its properties are summarized in Table~\ref{tab:phantom-table}. 
%
\end{remarknum}
\color{black}\begin{exmp}\label{exmp:phantom_four_users}
For $K=4$, $L=2$, $\gamma=1/2$ and
$(K_1,G_1)=(2,1)$, $(K_2,G_2)=(2,2)$, consider the
feasible $\mathrm{opt}$ reference point
$(\hat\Omega_{\mathrm{opt},1},\hat\beta_{\mathrm{opt},1})=(3,2)$.
The placement has $\binom42=6$ packets per file, followed by
a delivery split $\hat\phi_{\mathrm{opt},1}=2$, so
$\hat\Theta_{\mathrm{opt},1}=12$.
The four MC intervals deliver $4+4+5+5=18$ useful
subpackets, and the residual counts are
$(d_1,d_2,d_3,d_4)=(3,3,0,0)$. Hence
$Z_{\mathrm{opt}}^{\mathrm{Tot}}=24$,
$Z_{\mathrm{opt}}^{\mathrm{UC}}=6$ and
$S_{\mathrm{opt}}^{\mathrm{UC}}=\max\{\lceil6/2\rceil,
\lceil3/1\rceil\}=3$.
With $I=1$ and $\epsilon_{\mathrm{opt},1}=\zeta_{\mathrm{opt},1}=1$,
Theorem~\ref{thm:phantom_dof} gives
$\mathrm{DoF}_{\mathrm{opt},\mathrm{Ph}}(3,2)
=24/(4+3)=24/7$.
This example verifies the selected reference point, not a maximization
over all end-to-end Phantom operating points.
\end{exmp}
%
When all users have the same receive dimension $G$, choosing
$\hat G_1=G$ retains every stream, so $I=1$ and $Z_{\pi}^{\rm UC}=0$;
Theorem~\ref{thm:phantom_dof} then recovers exactly the selected
symmetric reference-policy DoF for each
$\pi\in\{\mathrm{opt},\mathrm{cmb},\mathrm{lin}\}$.\color{black}

\color{black}

\begin{exmp}
\label{examp:phantom}
Here, we briefly demonstrate the operating procedure of the Phantom strategy for the $\mathrm{opt}$ policy 
by considering an asymmetric MIMO-CC setup with $L=16$, $\gamma=0.04$, and $J=3$ user groups. The user groups have $\{K_{1},K_{2},K_{3}\}=\{25,75,125\}$ users, and the spatial multiplexing gains of users in each group are given as $\{G_{1},G_{2},G_{3}\}=\{2,4,8\}$.
For this network, the Phantom strategy concludes after $I=2$ MC rounds and one UC round. In the first MC round ($i=1$), the set of users is given as $\hat \CK_1=\CK=\{1,2,\ldots,225\}$. Considering a hypothetical symmetric MIMO-CC setup with $\hat G_1=8$ receive antennas per user and using the reference policy $\mathrm{opt}$, for the number of served users per transmission and the number of assigned streams per selected user we get $\hat \Omega_{\mathrm{opt},1}^*=13$ and $\hat \beta_{\mathrm{opt},1}^*=4$, respectively. Accordingly, users in groups~2 and~3 can decode all their parallel streams, while $\hat \beta_{\mathrm{opt},1}^*-\SfG_1=2$ streams should be removed randomly for users in group~1 to ensure linear decodability.
Following this process, in the MC round~2, we serve $\hat \CK_2=\CK_1=\{1,2,\ldots,25\}$. The parameters for the $\mathrm{opt}$ policy are calculated as 
$\hat \Omega_{\mathrm{opt},2}^*=17$ 
and $\hat \beta_{\mathrm{opt},2}^*=2$,  where $\hat \Omega_{\mathrm{opt},2}^*$ can be selected freely from the given range. As $\SfG_1 \ge \hat \beta_{\mathrm{opt},2}^*$, there is no need to \color{black}remove \color{black} any streams anymore, and the MC rounds are concluded. It can be easily verified that, after counting for the UC round, the final DoF for this setup is calculated as $\mathrm{DoF}_{\mathrm{opt},\mathrm{Ph}}=44.05$. 
\end{exmp}
\subsection{Complexity--Performance Tradeoffs and Finite-SNR Operation} \label{subsec:tradeoffs_finite_snr} The considered policies and asymmetric delivery strategies %
offer different tradeoffs among achievable DoF, subpacketization, and implementation complexity.
For each asymmetric policy--strategy pair, a symmetry-based reference upper bound is obtained by applying the same policy to the corresponding enhanced symmetric system in which every user has receive dimension $G_J$. For example, for the $\mathrm{opt}$ policy, such an upper bound can be found using $G=G_J$ in~\eqref{eq:total_DoF}. The $\mathrm{opt}$ policy provides the most flexible achievable-DoF benchmark, although its combinatorial subpacketization may limit its practical use in large systems. The $\mathrm{cmb}$ policy offers a more structured intermediate alternative, whereas $\mathrm{lin}$ is particularly attractive within its applicability range because of its linear subpacketization. 

The asymmetric strategies introduce a complementary tradeoff. \emph{Phantom} provides a flexible high-performance construction, potentially at the cost of higher subpacketization, while SPh offers a lower-complexity alternative. \emph{Grouping} and \emph{Super-grouping} reduce the effective combinatorial problem size by decomposing a large heterogeneous delivery problem into smaller group-level constructions. Selected policy--strategy combinations may also be implemented using shared-cache MIMO~\cite{naseritehrani2026low} or MIMO-PDA-based schemes~\cite{huang2026placement}, providing lower-subpacketization realizations of some operating points under the corresponding feasibility and divisibility conditions. 

Each achievable DoF operating point also defines a candidate finite-SNR transmission mode through its scheduled user set and stream allocation. These modes can be compared for a given SNR and channel realization based on their optimized symmetric rates. At lower SNR, more conservative modes, such as \emph{min-$G$} and \emph{Super-grouping}, may be preferable because reducing the number of simultaneously transmitted streams can improve the symmetric rate. As the SNR increases, \emph{Grouping} can benefit from greater spatial multiplexing, while \emph{Super-grouping} and \emph{Phantom} may further exploit coded-multicasting opportunities. Once a mode is selected, the finite-SNR beamforming and receive-combining methods in~\cite{naseritehrani2026cache} can be applied to its stream allocation.

Beyond subpacketization, implementation complexity depends on the numbers of multicast and residual transmission intervals, which determine the scheduling and signaling overhead and the number of beamforming instances required. CSI-dependent beamformers must also be updated when delivery spans multiple channel-coherence blocks. For a fixed system configuration, however, the multicast structure, stream allocation, and packet-index schedule can be precomputed offline and reused across demand realizations, with only the requested packet indices updated online. Receivers buffer the decoded subpackets until their requested files can be reconstructed. The computational cost of each beamformer depends on the chosen implementation and is separate from the delivery-level complexity considered here.  

\color{black}
\begin{table*}[t]
    \centering
    \caption{The \emph{Single-round Phantom} ($\mathrm{SPh}$) strategy parameters for different reference policies \color{black} and using normalized UC approximation in~\eqref{eq:DoF_phantom_general} under the condition~\eqref{eq:uc_rxa_activation}\color{black}. Note that setting $\hat I=1$ gives $\epsilon_{\pi,1}=\zeta_{\pi,1}=1$; hence the parameters reduce to $(\CK,K,\hat G,\hat\Omega_{\pi},\hat\beta_{\pi}, \bm \beta_{\pi}^k)$ after dropping the round index.}
    \label{tab:phantom-table}
    \begin{tabular}{>{\centering\arraybackslash}p{0.12\textwidth}!{\vrule width 1pt}>{\centering\arraybackslash}p{0.26\textwidth}|>{\centering\arraybackslash}p{0.26\textwidth}|>{\centering\arraybackslash}p{0.26\textwidth}}
        \textbf{Policy} & \textbf{DoF-optimized} & \textbf{Combinatorial parallel links} & \textbf{Cyclic parallel links} \\
        \Xhline{1pt}
        \textbf{DoF} 
        & $\begin{aligned}\\[-5mm]
        &\textrm{DoF}_{\mathrm{opt}, \mathrm{SPh}}\approx \\
        &\frac{\displaystyle   L }{\displaystyle  1 - \frac{
     \displaystyle 
     1}{\displaystyle K \hat{\beta}_{\rm opt}^{*}
} \displaystyle \big(\sum\limits_{k\in \CK}\bm\beta_{\rm opt}^k - \frac{LK}{\hat{\Omega}_{\rm opt}^{*}}\big) } \end{aligned}$
        & $\begin{aligned}\\[-4mm]
        &\mathrm{DoF}^{}_{\mathrm{cmb},\mathrm{SPh}} \approx \\
        &\frac{\displaystyle   L }{\displaystyle  1 - \frac{
     \displaystyle 
     1}{\displaystyle K \hat{G}
} \displaystyle \big(\sum\limits_{k\in \CK}\bm\beta_{\mathrm{cmb}}^k - \frac{LK}{\left(\! K\gamma + \big\lfloor \tfrac{L}{\hat{G}} \big\rfloor\! \right)}\big) }  
        \end{aligned}$ 
        & $\begin{aligned}\\[-4mm]
        &\mathrm{DoF}^{}_{\mathrm{lin},\mathrm{SPh}} \approx \\
        &\frac{\displaystyle   L }{\displaystyle  1 - \frac{
     \displaystyle 
     1}{\displaystyle K \hat{G}
} \displaystyle \big(\sum\limits_{k\in \CK}\bm\beta_{\mathrm{lin}}^k - \frac{LK}{\left(\! K\gamma + \big\lfloor \tfrac{L}{\hat{G}} \big\rfloor\! \right)}\big) } 
        \end{aligned}$ \\ 
        \hline
        \makecell{
        \textbf{DoF} \\
        per multicast
        }
        &$\begin{aligned}\\[-3mm]
        &\hat{\textrm{DoF}}_{\mathrm{opt}} = \hat{\Omega}_{\rm opt}^{*} \cdot \hat{\beta}_{\rm opt}^{*} \\[0.2mm]
        &\hat{\Omega}_{\rm opt}^{*} , \hat{\beta}_{\rm opt}^{*} \: \leftarrow \Delta\Big(K,\gamma,L,\! \hat{G}\Big)
        \end{aligned}$ 
        & $\begin{aligned}
        &\hat{\mathrm{DoF}}_{\mathrm{cmb}} = \\
        & \qquad\qquad\hat{G}\left( K\gamma + \big\lfloor \tfrac{L}{\hat{G}} \big\rfloor \right)
        \end{aligned}$
        & $\begin{aligned}
        &\mathrm{DoF}^{}_{\mathrm{lin}} = \\
        & \qquad\qquad \hat{G}  \:\times \: \left(\! K\gamma + \big\lfloor \tfrac{L}{\hat{G}} \big\rfloor\! \right)
        \end{aligned}$
        \\
        \hline
        \makecell{
        \textbf{Subpacketization} \\
        }
        & $\begin{aligned}\\[-4mm]
        &\hat{\Theta}^{}_{\mathrm{opt}} = \binom{K}{K\gamma} \binom{K-K\gamma-1}{\hat{\Omega}_{\rm opt}^{*}-K\gamma-1} \hat{\beta}_{\rm opt}^{*}
        \end{aligned}$
        & $\begin{aligned}\\[-3mm]
        &\hat{\Theta}^{}_{\mathrm{cmb}} = \binom{K}{K\gamma} \binom{K-K\gamma-1}{\big\lfloor \tfrac{L}{\hat{G}} \big\rfloor-1} \hat{G}_{}
        \end{aligned}$
        & $\begin{aligned}\\[-4mm]
        &\hat{\Theta}^{}_{\mathrm{lin}} = \hat{G} K \: \left(\! K\gamma + \big\lfloor \tfrac{L}{\hat{G}} \big\rfloor\! \right)
        \end{aligned}$ \\
        \hline
        \makecell{
        \textbf{Trans. count} \\
        MC and UC
        }
        & $\begin{aligned}\\[-7mm]
        &{S}_{\mathrm{opt}}^{\textrm{MC}} =  \binom{K}{\hat{\Omega}_{\rm opt}^{*}} \binom{\hat{\Omega}_{\rm opt}^{*}-1}{K\gamma}\\[3mm]
        &{S}_{\mathrm{opt}}^{\textrm{UC}} = \frac{1}{L}\binom{K-1}{\hat{\Omega}_{\rm opt}^{*}-1} \binom{\hat{\Omega}_{\rm opt}^{*}-1}{K\gamma} \times\\&
     \qquad \sum\limits_{k\in \CK}\Big(\hat{\beta}_{\rm opt}^{*}-\bm\beta_{\mathrm{opt}}^k\Big)
        \end{aligned}$ 
        & $\begin{aligned}\\[-3.6mm]
        &S_{\mathrm{cmb}}^{\textrm{MC}} = \\[-1mm]
        & \quad\binom{K}{K\gamma+\big\lfloor \tfrac{L}{\hat{G}} \big\rfloor}\binom{K\gamma+\big\lfloor \tfrac{L}{\hat{G}} \big\rfloor-1}{K\gamma}\\[1mm]
        &S_{\mathrm{cmb}}^{\textrm{UC}} = \frac{1}{L}\binom{K-1}{K\gamma+\big\lfloor \tfrac{L}{\hat{G}} \big\rfloor-1} \times \\&\hspace{-2mm}\binom{K\gamma+\big\lfloor \tfrac{L}{\hat{G}} \big\rfloor-1}{K\gamma}\!\!\sum\limits_{k\in \CK}\Big(\hat{G}-\bm\beta_{\mathrm{cmb}}^k\Big) 
        \end{aligned}$ 
        & $\begin{aligned}\\[-9mm]
        &S_{\mathrm{lin}}^{\textrm{MC}} =  K\big(K-K\gamma\big)\\[6mm]
        &S_{\mathrm{lin}}^{\textrm{UC}} = \frac{1}{L}
     \displaystyle (K-K\gamma) \times
    \\&\quad  \!\!\left(\! K\gamma + \big\lfloor \tfrac{L}{\hat{G}} \big\rfloor\! \right)\:
    \sum\limits_{k\in \CK}\Big(\hat{G}-\bm\beta_{\mathrm{lin}}^k\Big) 
        \end{aligned}$ \\
        \hline
        \textbf{Applicability} 
        & All network parameters 
        & All network parameters 
        & $\displaystyle\big\lfloor \displaystyle\tfrac{\displaystyle L}{ \displaystyle \hat{G}} \big\rfloor \ge \displaystyle K
\gamma,
\quad $
\\
    \end{tabular}%
\end{table*}

\section{Numerical Analysis}
\label{section:Simulations}


Numerical experiments 
evaluate the performance of the proposed asymmetric MIMO-CC strategies. \color{black}
%
Unless otherwise stated, the \emph{Phantom}
results use the BS-side UC specialization
of Theorem~\ref{thm:phantom_dof}.
For $\mathrm{lin}$, we use the largest feasible
stream count and the corresponding maximal
$\Omega$.
\color{black}

Fig.~\ref{fig:Plot_scaling_K2} illustrates the impact of user heterogeneity on the achievable DoF under different strategies and policies for a network of $K=35$ users, divided into $J=2$ groups of sizes $K_1 = 35-K_2$ and $K_2$, where $K_2\in\{5,10,\ldots,30\}$. \color{black} At $K_2=30$, the general UC completion is used for the $\mathrm{opt}$ and
$\mathrm{cmb}$ \emph{Phantom}. \color{black} The receive antenna configuration for these two groups is set to $G_1 = 2$ and $G_2 = 8$. For comparison, we have also included the achievable DoF of the multi-user MIMO (MU-MIMO) baseline, where cache memories are used for a local caching gain only. As observed, for the policy $\pi \in \{\mathrm{opt}, \mathrm{cmb}\}$,
both the \emph{Grouping} and \emph{Phantom} strategies effectively translate the increasing share of users with more receive antennas
(in set $K_2$) into a clear DoF gain, across the entire range of $K_2$ (in contrast, the $\mathrm{lin}$ policy is structurally constrained in this setup due to its applicability constraints). This highlights how well the proposed strategies can integrate the available Rx spatial multiplexing gain, which is non-uniformly distributed across user sets, with the global caching gain. 
Moreover, all proposed strategies outperform the 
MU-MIMO baseline across the full range of parameters, confirming that caching and spatial multiplexing gains
remain complementary even under pronounced asymmetry.  
Table~\ref{orderofComplexity} collects the corresponding subpacketization levels, highlighting the inherent performance--complexity tradeoff across different policies and
strategies.

\begin{figure}[t]
    \centering

    \resizebox{.8\columnwidth}{!}{%

    \begin{tikzpicture}

    \begin{axis}
    [
    axis lines = center,
    xlabel near ticks,
    xlabel = \smaller {$K_{2}$},
    ylabel = \smaller {Achievable DoF},
    ylabel near ticks,
    ymin =9,
    ymax = 50,
    xmin = 5,  
    xmax = 30,
     xticklabel style={/pgf/number format/fixed}, 
    legend style={
    nodes={scale=0.7, transform shape},
    at={(.34,1.0)}, 
    anchor=north, 
    draw=black,
    outer sep=2pt,
    font=\normalsize,
    legend columns=2, 
    scale=0.7, 
    text height=3.0ex, text depth=.25ex 
    },
    ticklabel style={font=\smaller},
    grid=both,
    minor tick num = 4, 
    major grid style = {line width=0.5pt,draw=gray!50},
    minor grid style = {line width=0.2pt,draw=gray!30}
    ]


    \addplot
    [dash dot, mark = square, mark size=3.6,mark options={solid, fill=black!50} , black!100]
    table[y=MUMIMO,x=K2]{Figs/DoF_Prop_T1_1.tex};
    \addlegendentry{\normalsize MU-MIMO };
    
    \addplot
    [dash dot, mark = triangle, mark size=3.2,mark options={solid, fill=blue!100} , blue!100]
    table[y=minG,x=K2]{Figs/DoF_Prop_T1_1.tex};
    \addlegendentry{\normalsize \textit{min-$G$}
    };


    \addplot
    [dashed, line width= 1pt, mark = x, mark size=3.6,mark options={solid, fill=red!100} , red!100]
    table[y=phantom,x=K2]{Figs/DoF_Prop_T1_2.tex};
    \addlegendentry{\normalsize \textit{Phantom}, $\mathrm{opt}$ };

    \addplot
    [mark = o , mark size=3.4, green!100]
    table[y=grouping,x=K2]{Figs/DoF_Prop_T1_1.tex};
    \addlegendentry{\normalsize \textit{Grouping}, $\mathrm{opt}$ };


   \addplot
    [dashed, line width= 1pt, mark = +, mark size=3.6,mark options={solid, fill=cyan!100} , magenta!100]
    table[y=phantomcpar,x=K2]{Figs/DoF_Prop_T1_2.tex};
    \addlegendentry{\normalsize \textit{Phantom}, $\mathrm{cmb}$ };

    \addplot
    [mark = * , mark size=3.4, cyan!100]
    table[y=groupingcpar,x=K2]{Figs/DoF_Prop_T1_1.tex};
    \addlegendentry{\normalsize \textit{Grouping}, $\mathrm{cmb}$ };
    

      \addplot
    [dashed, mark = pentagon, mark size=3.7,mark options={solid, fill=orange!100} , orange!100]
    table[y=phantomlpar,x=K2]{Figs/DoF_Prop_T1_1.tex};
    \addlegendentry{\normalsize \textit{Phantom}, $\mathrm{lin}$ };

    \addplot
    [mark = 10-pointed star , mark size=3.6, olive!100]
    table[y=groupinglpar,x=K2]{Figs/DoF_Prop_T1_1.tex};
    \addlegendentry{\normalsize \textit{Grouping}, $\mathrm{lin}$ };
    
    
    

    \end{axis}

    \end{tikzpicture}
    }
    \vspace{-2mm}
    \caption{User heterogeneity impact on DoF scaling under different strategies and policies. $L=14$, $\gamma = 0.2$, and $K_1 = 35-K_2$. 
    }
    \label{fig:Plot_scaling_K2}
\end{figure}
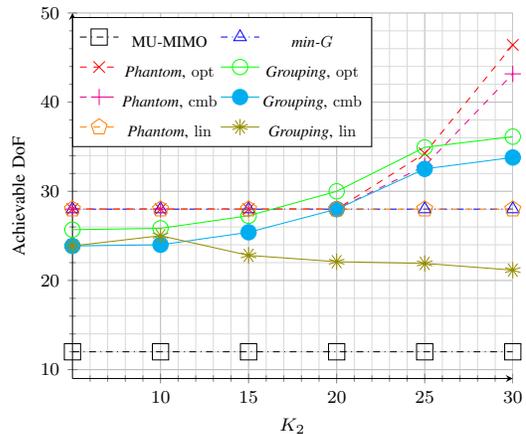

\begin{table}[!t]
\centering
\caption{Subpacketization order of different strategies, network setup similar to Fig.~\ref{fig:Plot_scaling_K2}, $K_1 = 5$, $K_2 = 30$.}
\label{orderofComplexity}
\vspace{-0mm}
\renewcommand{\arraystretch}{1.31}
\fontsize{16}{16}\selectfont
\vspace{-0mm}
\resizebox{.4\textwidth}{!}{
\begin{tabular}{|c|*{3}{>{\centering\arraybackslash}m{3.5cm}|}}
\hline
\diagbox[dir=NW,width=7.9em,height=2.5em]{\fontsize{20}{22}\selectfont Strategy}{\fontsize{20}{20}\selectfont Policy} 
  & \fontsize{20}{20}\selectfont $\mathrm{opt}$ 
  & \fontsize{20}{20}\selectfont $\mathrm{cmb}$ 
  & \fontsize{20}{20}\selectfont $\mathrm{lin}$ \\
\hline
\fontsize{20}{20}\selectfont \emph{min-$G$}
  & \fontsize{20}{20}\selectfont $3.98\times 10^{12}$
  & \fontsize{20}{20}\selectfont $3.98\times 10^{12}$
  & \fontsize{20}{20}\selectfont $9.8\times 10^{2}$ \\
\hline
\fontsize{20}{20}\selectfont \emph{Grouping}
  & \fontsize{20}{20}\selectfont $1.1\times 10^{8}$
  & \fontsize{20}{20}\selectfont $4.75\times 10^{6}$
  & \fontsize{20}{20}\selectfont $8.3\times 10^{2}$ \\
\hline
\fontsize{20}{20}\selectfont \emph{Phantom}
  & \fontsize{20}{20}\selectfont $1.45\times 10^{9}$
  & \fontsize{20}{20}\selectfont $5.37\times 10^{7}$
  & \fontsize{20}{20}\selectfont $9.8\times 10^{2}$ \\
\hline
\end{tabular}
}
\end{table}
In Table~\ref{effgrps}, the performance of the \emph{Super-grouping} strategy is evaluated for $J=5$ user sets. 
The results show that \emph{Super-grouping} can constructively enhance the achievable DoF of the asymmetric MIMO-CC system compared to both \emph{min-$G$} ($\bar{J}=1$) and \emph{Grouping} ($\bar{J}=J=5$) strategies. 
The maximum DoF is attained when $\bar{J}=3$, with the equivalent user set configuration $\{\{1\}, \{2,3,4\}, \{5\}\}$ (which is one of the $\binom{5-1}{3-1}=6$ possible combinations for $\bar{J}=3$). This increased DoF results from an optimized trade-off in~\eqref{grp_delivery34_opt_concise}, between receiver-side spatial multiplexing and CC gains in equivalent user sets. 

\begin{table}[!t]
\centering
\caption{Achievable DoF of the \emph{Super-grouping} strategy for all the possible combinations of user sets (shown in red color),
$J=5$, $L=16$, $\gamma = 0.05$, $\{K_j\}_{j=1}^5 = \{20,20,20,20,220\}$, and $\{G_j\}_{j=1}^5 = \{2,5,6,7,16\}$.
}\label{effgrps}
\vspace{-0mm}
\renewcommand{\arraystretch}{1.3}
\fontsize{14}{29}\selectfont
\vspace{-0mm}
\resizebox{.485\textwidth}{!}{
\begin{tabular}{|c|c|*{6}{>{\centering\arraybackslash}m{3.75cm}|}}
\hline
\fontsize{24}{20}\selectfont $\bar{J}$ & \fontsize{24}{20}\selectfont $\binom{J-1}{\bar{J}-1}$ & \multicolumn{6}{|c|}{\fontsize{24}{20}\selectfont$\textrm{DoF}_{\mathcal{\bar{J}}}$} \\
\hline
\fontsize{24}{20}\selectfont 1 & \fontsize{24}{20}\selectfont 1 & \makecell{\fontsize{24}{20}\selectfont 46 \\[-3mm] \textcolor{red}{\{1,2,3,4,5\}}} & -- & -- & -- & -- & -- \\
\hline
\fontsize{24}{20}\selectfont 2 & \fontsize{24}{20}\selectfont 4 & \makecell{\fontsize{24}{20}\selectfont 71.0526 \\[-3mm] \textcolor{red}{\{1\},\{2,3,4,5\}}} & \makecell{\fontsize{24}{20}\selectfont 63.7168 \\[-3mm] \textcolor{red}{\{1,2\},\{3,4,5\}}} & \makecell{\fontsize{24}{20}\selectfont 59.8446 \\[-3mm] \textcolor{red}{\{1,2,3\},\{4,5\}}} & \makecell{\fontsize{24}{20}\selectfont 66.9767 \\[-3mm] \textcolor{red}{\{1,2,3,4\},\{5\}}} & -- & -- \\
\hline
\fontsize{24}{20}\selectfont 3 & \fontsize{24}{20}\selectfont 6 & \makecell{\fontsize{24}{20}\selectfont 62.8690 \\[-3mm] \textcolor{red}{\{1\},\{2\},\{3,4,5\}}} & \makecell{\fontsize{24}{20}\selectfont 63.4228 \\[-3mm] \textcolor{red}{\{1\},\{2,3\},\{4,5\}}} & \makecell{\fontsize{24}{20}\selectfont 75.5433 \\[-3mm] \textcolor{red}{\{1\},\{2,3,4\},\{5\}}} & \makecell{\fontsize{24}{20}\selectfont 58.6047 \\[-3mm] \textcolor{red}{\{1,2\},\{3\},\{4,5\}}} & \makecell{\fontsize{24}{20}\selectfont 66.9767 \\[-3mm] \textcolor{red}{\{1,2\},\{3,4\},\{5\}}} & \makecell{\fontsize{24}{20}\selectfont 63.7425 \\[-3mm] \textcolor{red}{\{1,2,3\},\{4\},\{5\}}} \\
\hline
\fontsize{24}{20}\selectfont 4 & \fontsize{24}{20}\selectfont 4 & \makecell{\fontsize{24}{20}\selectfont  57.8867 \\[-3mm] \textcolor{red}{\{1\},\{2\},\{3\},\{4,5\}}} & \makecell{\fontsize{24}{20}\selectfont 66.0406 \\[-3mm] \textcolor{red}{\{1\},\{2\},\{3,4\},\{5\}}} & \makecell{\fontsize{24}{20}\selectfont 67.8179 \\[-3mm] \textcolor{red}{\{1\},\{2,3\},\{4\},\{5\}}} & \makecell{\fontsize{24}{20}\selectfont 62.3377 \\[-3mm] \textcolor{red}{\{1,2\},\{3\},\{4\},\{5\}}} & -- & -- \\
\hline
\fontsize{24}{20}\selectfont 5 & \fontsize{24}{20}\selectfont 1 & \makecell{\fontsize{24}{20}\selectfont 61.5259 \\[-3mm] \textcolor{red}{\{1\},\{2\},$\cdots$,\{5\}}} & -- & -- & -- & -- & -- \\
\hline
\end{tabular}
}
\end{table}

Table~\ref{Table_1} illustrates the operating procedure of the \emph{Single-round Phantom} ($\mathrm{SPh}$) strategy and its achievable DoF values under the $\mathrm{opt}$ policy, for the network setup considered in Example~\ref{examp:phantom}, i.e., $J=3$ and $\{G_j\}_{j=1}^3 = \{2,4,8\}$.
As observed, except for the special case when $\hat{G}=\check{G}$ and $\hat{\Omega}=\Omega^*_{\check{G}}$ where
\emph{min-$G$} and $\mathrm{SPh}$
attain the same DoF, 
the $\mathrm{SPh}$ strategy provides a larger DoF boost than both the \emph{Grouping} and \emph{min-$G$} strategies, as it simultaneously exploits the maximum available CC gain and the receive-side spatial multiplexing capabilities of all users. Moreover, the achievable DoF of the $\mathrm{SPh}$ strategy can be fine-tuned through the choice of \(\hat{G}\) and \(\Omega\), while strictly adhering to the linear decodability condition in Lemma~\ref{LD_Phantom}.

\begin{table}[!t]
\vspace{10pt}
\centering
\caption{Operating procedure and achievable DoF of the $\mathrm{SPh}$ strategy, and comparison with \emph{min-$G$} and \emph{Grouping} strategies, $J=3$, $\{K_j\}_{j=1}^3 = \{25,75,125\}$, $\{G_j\}_{j=1}^3 = \{2,4,8\}$, $L=16$, and $\gamma=0.04$.}
\resizebox{.49\textwidth}{!}{
\vspace{-1mm}
\renewcommand{\arraystretch}{1.3}
\fontsize{22}{24}\selectfont
\vspace{-2mm}

\resizebox{\textwidth}{!}{
\begin{tabular}{|c|c|c|c|c|c|c|c|c|c|c|c|c|}
\hline
\multicolumn{9}{|c|}{$\textrm{DoF}_{\mathrm{opt},\mathrm{SPh}}$} & $\textrm{DoF}^{*}_{\check{G}}$ & \multicolumn{3}{|c|}{$\textrm{DoF}_{\CJ}^*$= 35.72} \\
\hline
\diagbox[dir=NW,width=7.5em, height=2em]{$\hat{\beta}$ \text{or} $\beta_k$}{\raisebox{-.3em}{$\Omega$}} & 10 & 11 & 12 & 13 & 14 & 15 & 16 & 17 & 17 & 9 & 7 & 7 \\
\hline
2 & 20 & 22 & 24 & 26 & 28 & 30 & 32 & 34 & 34 & 18 & -- & -- \\
3 & 28.05 & 30.66 & 33.23 & 35.77 & 38.28 & 40.75 & -- & -- & -- & -- & -- & -- \\
4 & 35.12 & 38.17 & 41.14 & 44.05 & -- & -- & -- & -- & -- & -- & 28 & -- \\
5 & 35.29 & 37.71 & 40 & 42.16 & -- & -- & -- & -- & -- & -- & -- & -- \\
6 & 35.41 & 37.42 & 39.27 & -- & -- & -- & -- & -- & -- & -- & -- & -- \\
7 & 35.49 & 37.21 & 38.77 & -- & -- & -- & -- & -- & -- & -- & -- & -- \\
8 & 35.56 & 37.05 & -- & -- & -- & -- & -- & -- & -- & -- & -- & 56 \\
\hline
\end{tabular}\label{aTable}
}

\label{Table_1}
}
\end{table}

\begin{table}[t]
    \caption{Simulations setups for Fig.~\ref{fig:plot_prim_over_Imp}, Fig.~\ref{fig:plot_lpar_opt_comparison}, and Fig.~\ref{fig:plot_scal_Rx_SM}.}
    \centering
    \resizebox{\columnwidth}{!}{
    \begin{tabular}{|c|c|c|c|c|c|}
        \hline
        Setup & $L$ & $\gamma$ & $J$ & $\{K_j\}_{j=1}^J$ & $\{G_j\}_{j=1}^J$ \\
        \hline
        $\mathrm{sa1}$ & 28 & 0.05 & 2 & $\{20,60\}$ & $\{2,4\}$ \\
        $\mathrm{sa2}$ & 16 & 0.05 & 3 & $\{20,40,100\}$ & $\{3,5,7\}$ \\
        $\mathrm{sa3}$ & 16 & 0.05 & 4 & $\{20,20,40,100\}$ & $\{4,8,12,16\}$ \\
        $\mathrm{sa4}$ & 16 & 0.05 & 3 & $\{60,100,200\}$ & $\{2,6,8\}$ \\
        $\mathrm{sa5}$ & 16 & 0.1 & 3 & $\{60,100,200\}$ & $\{2,6,8\}$ \\
        \hline
        $\mathrm{sb1}$ & 28 & 0.05 & 2 & $\{40,100\}$ & $\{2,4\}$ \\
        $\mathrm{sb2}$ & 16 & 0.01 & 3 & $\{100,100,200\}$ & $\{2,4,8\}$ \\
        $\mathrm{sb3}$ & 16 & 0.01 & 4 & $\{50,50,100,100\}$ & $\{3,5,7,9\}$ \\
        $\mathrm{sb4}$ & 24 & 0.04 & 3 & $\{25,100,25\}$ & $\{2,4,16\}$ \\
        \hline
        $\mathrm{sc1}$ & 16 & 0.04 & 3 & $\{150,25,50\}$ & $\{2,4,8\}$ \\
        $\mathrm{sc2}$ & 16 & 0.04 & 3 & $\{50,150,25\}$ & $\{2,4,8\}$ \\
        $\mathrm{sc3}$ & 16 & 0.04 & 3 & $\{50,25,150\}$ & $\{2,4,8\}$ \\
        $\mathrm{sc4}$ & 16 & 0.04 & 3 & $\{25,50,150\}$ & $\{2,4,8\}$ \\
        $\mathrm{sc5}$ & 16 & 0.05 & 5 & $\{20,20,20,40,100\}$ & $\{2,4,6,8,10\}$ \\
        $\mathrm{sc6}$ & 16 & 0.05 & 5 & $\{20,20,20,40,100\}$ & $\{4,6,8,10,12\}$ \\
        $\mathrm{sc7}$ & 16 & 0.05 & 5 & $\{20,20,20,40,100\}$ & $\{2,10,12,14,16\}$ \\
        \hline
    \end{tabular}
    }
    \label{tab:sim_setups}
\end{table}

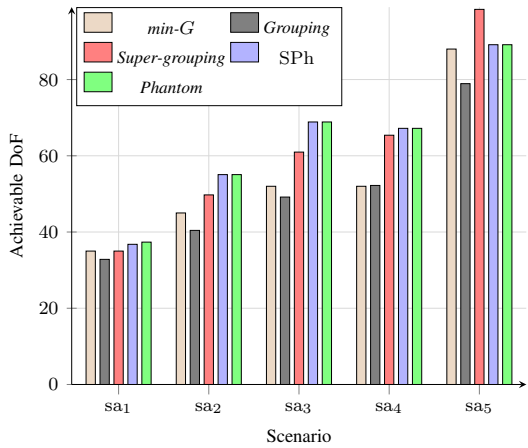
\begin{figure}[t]
    \centering
    \resizebox{.8\columnwidth}{!}{%

   
    \begin{tikzpicture}
\begin{axis}[
    axis lines = left,
    xlabel = \smaller {Scenario},
    ylabel = \smaller {Achievable DoF},
    ylabel near ticks,
    xlabel near ticks,
    legend pos = north west,
    legend columns=2,
    legend style = {
        nodes={scale=0.99, transform shape},
        cells={align=center},
        at={(0.01,1)},
        anchor=north west,
        font=\small
    },
    ybar,
    legend image code/.code={%
      \draw[#1] (0cm,-0.045cm) rectangle (0.42cm,0.25cm);
    },
    ymin = 0,
    ymax =99,
    xtick=data,
    enlarge x limits=0.13, 
    symbolic x coords={sp1,sp2,sp3,sp4,sp5},
    xtick=data,
    xticklabels={$\mathrm{sa}_1$,$\mathrm{sa}_2$,$\mathrm{sa}_3$,$\mathrm{sa}_4$,$\mathrm{sa}_5$},
    ticklabel style={font=\smaller},
    grid=both,
    major grid style={line width=.1pt,draw=gray!30},
    bar width=0.14cm,
]
 \addplot 
    [black,fill=brown!30]
    table
    [x=Scen ,y=minGsbopt]
    {Figs/DoF_imp_over_prim1.tex};
    \addlegendentry{\smaller \textit{min-$G$}}

    \addplot 
    [black,fill=black!50]
    table
    [x=Scen ,y=groupingsbop]
    {Figs/DoF_imp_over_prim1.tex};
    \addlegendentry{\smaller \textit{Grouping}}
    
    \addplot 
    [black,fill=red!50]
    table
    [x=Scen ,y=groupingsbopeff]
    {Figs/DoF_imp_over_prim1.tex};
    \addlegendentry{\smaller \textit{Super-grouping}}

    \addplot 
    [black,fill=blue!30]
    table
    [x=Scen ,y=phantomsbopt]
    {Figs/DoF_imp_over_prim1.tex};
    \addlegendentry{\smaller $\mathrm{SPh}$}

    \addplot 
    [black,fill=green!50]
    table
    [x=Scen ,y=phantomsimpbopt]
    {Figs/DoF_imp_over_prim1.tex};
    \addlegendentry{\smaller \textit{Phantom}}
    


    \end{axis}
    \end{tikzpicture}
    }
    \vspace{-2mm}
    \caption{DoF comparison, hybrid vs primary strategies.
    }
    \label{fig:plot_prim_over_Imp}
\end{figure}

In Fig.~\ref{fig:plot_prim_over_Imp}, the proposed hybrid strategies, i.e., \emph{Super-grouping}, \emph{Phantom}, and $\mathrm{SPh}$, are evaluated under the $\mathrm{opt}$ policy, and their achievable DoF is compared with the primary strategies, i.e., \emph{min-$G$} and \emph{Grouping} (simulation setups $\mathrm{sa}_1$--$\mathrm{sa}_5$ are defined in Table~\ref{tab:sim_setups}). 
We observe that the hybrid strategies consistently outperform their primary counterparts by effectively adapting their multicast design to the
heterogeneous antenna profile. 
The $\mathrm{SPh}$ strategy closely tracks the performance of the \emph{Phantom} strategy in all cases, showing that its simplified construction preserves the achievable DoF. 
Also, scenario $\mathrm{sa}_5$ illustrates that \emph{Super-grouping} can surpass \emph{Phantom} when the Rx-side spatial multiplexing gains, user-set sizes, and local caching gains are sufficiently large.

\begin{figure}[t]
    \centering
   \resizebox{.8\columnwidth}{!}{%
    \begin{tikzpicture}
    \begin{axis}[
        axis lines = left,
        xlabel = \smaller {Scenario},
        ylabel = \smaller {Achievable DoF},
        ylabel near ticks,
        xlabel near ticks,
        legend pos = north west,
        legend columns=2,
        legend style = {
            nodes={scale=0.99, transform shape},
            cells={align=center},
            at={(0.01,1)},
            anchor=north west,
            font=\small
        },
        ybar,
        legend image code/.code={%
        \draw[#1] (0cm,-0.045cm) rectangle (0.42cm,0.25cm);
        },
        ymin = 0,
        ymax =65,
        xtick=data,
        enlarge x limits=0.1,
        symbolic x coords={sp1a,sp1b,sp2a,sp2b,sp3a,sp3b,sp4a,sp4b},
        xticklabels={$\mathrm{sb}_1$,$\mathrm{sb}_1$,$\mathrm{sb}_2$,$\mathrm{sb}_2$,$\mathrm{sb}_3$,$\mathrm{sb}_3$,$\mathrm{sb}_4$,$\mathrm{sb}_4$,$\mathrm{sb}_5$,$\mathrm{sb}_5$},
        ticklabel style={font=\smaller},
        grid=both,
        minor tick num = 4,
        major grid style={line width=.5pt,draw=gray!30},
        minor grid style = {line width=0.3pt,draw=gray!30},
        bar width=0.1cm,
    ]

    \addplot 
        [black,fill=brown!30]
        table
        [x=Scen ,y=minGsbopt]
        {Figs/DoF_effgrp_phantom_senseitivity.tex};
    \addlegendentry{\smaller \textit{min-$G$}}


    \addplot 
        [black,fill=black!100]
        table
        [x=Scen ,y=groupingsbopeff]
        {Figs/DoF_effgrp_phantom_senseitivity.tex};
    \addlegendentry{\smaller \textit{Super-grouping}}

     \addplot 
        [black,fill=black!50]
        table
        [x=Scen ,y=phantomsbopt]
        {Figs/DoF_effgrp_phantom_senseitivity.tex};
    \addlegendentry{\smaller $\mathrm{SPh}$}
    
    \addplot 
        [black,fill=red!50]
        table
        [x=Scen ,y=phantomsimpbopt]
        {Figs/DoF_effgrp_phantom_senseitivity.tex};
    \addlegendentry{\smaller \textit{Phantom}}
    \node[font=\tiny, anchor=south] at (axis cs:sp1a,44) {$\mathrm{opt}$};
    \node[font=\tiny, anchor=south] at (axis cs:sp1b,46) {$\mathrm{lin}$};
    \node[font=\tiny, anchor=south] at (axis cs:sp2a,31) {$\mathrm{opt}$};
    \node[font=\tiny, anchor=south] at (axis cs:sp2b,30) {$\mathrm{lin}$};
    \node[font=\tiny, anchor=south] at (axis cs:sp3a,32)  {$\mathrm{opt}$};
    \node[font=\tiny, anchor=south] at (axis cs:sp3b,29) {$\mathrm{lin}$};
    \node[font=\tiny, anchor=south] at (axis cs:sp4a,48) {$\mathrm{opt}$};
    \node[font=\tiny, anchor=south] at (axis cs:sp4b,49) {$\mathrm{lin}$};

    \end{axis}

    \end{tikzpicture}
    }
    \vspace{-2mm}
    \caption{
    Comparison of the $\mathrm{opt}$ and $\mathrm{lin}$ policies.
    }
    \label{fig:plot_lpar_opt_comparison}
\end{figure}

In Fig.~\ref{fig:plot_lpar_opt_comparison}, the achievable DoF of the $\mathrm{opt}$ policy is compared against the
$\mathrm{lin}$ policy, under various strategies (simulation setups $\mathrm{sb}_1$-$\mathrm{sb}_4$ are defined in Table~\ref{tab:sim_setups}).
Across the considered setups, the $\mathrm{lin}$ policy closely tracks the $\mathrm{opt}$, despite its inherent structural constraints. Moreover, under the \emph{Phantom} strategy, the
$\mathrm{lin}$ policy can even achieve a larger DoF
than $\mathrm{opt}$ and $\mathrm{cmb}$, due to its smaller subpacketization, which makes multi-round multicasting more efficient as discussed in Remark~\ref{s-r-ph}.

 In Fig.~\ref{fig:Tx_side_scaling}, the scaling behavior of the achievable DoF with respect to the transmitter-side spatial multiplexing gain is studied 
 for 
 different strategies, under the $\mathrm{opt}$ policy.
As $L$ increases, all strategies benefit from the enhanced transmitter side spatial multiplexing, albeit with different scaling behaviors. The \emph{Phantom} strategy consistently achieves the largest DoF and the steepest growth, followed by \emph{Super-grouping}, \emph{Grouping}, and
\emph{min-$G$}. This yields the ordering $\mathrm{DoF}_{\mathrm{opt},\check{G}}^* \leq \mathrm{DoF}_{\mathrm{opt},{\CJ}}^* \leq\mathrm{DoF}_{\mathrm{opt},\Bar{\CJ}}^* \leq \mathrm{DoF}_{\mathrm{opt},\mathrm{Ph}}^*$.
Moreover, while \emph{Grouping} is inferior to \emph{min-$G$} at small $L$, it surpasses \emph{min-$G$} once $L$ becomes sufficiently large ($L\approx 15$).
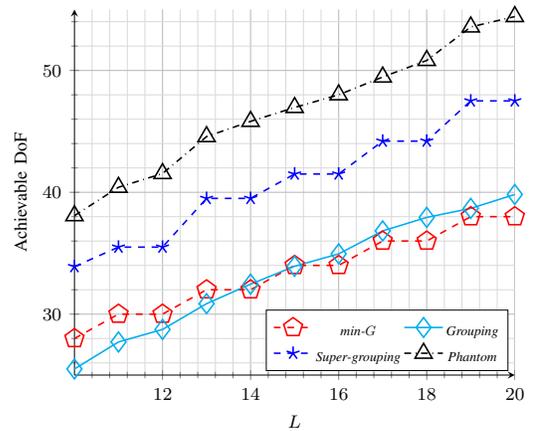
\begin{figure}[t]
    \centering

    \resizebox{.8\columnwidth}{!}{%

    \begin{tikzpicture}

    \begin{axis}
    [
    axis lines = center,
    xlabel near ticks,
    xlabel = \smaller {$L$},
    ylabel = \smaller {Achievable DoF},
    ylabel near ticks,
    ymin =25,
    ymax = 55,
    xmin = 10,  
    xmax = 20,
    legend style={
    nodes={scale=0.6, transform shape},
    at={(.71,0.19)}, 
    anchor=north, 
    draw=black,
    outer sep=2pt,
    font=\normalsize,
    legend columns=2, 
    scale=0.9, 
    text height=3.0ex, text depth=.25ex 
    },
    ticklabel style={font=\smaller},
    grid=both,
    minor tick num = 4, 
    major grid style = {line width=0.5pt,draw=gray!50},
    minor grid style = {line width=0.2pt,draw=gray!30}
    ]

   \addplot
    [dashed,line width=.7pt, mark = pentagon, mark size=4.4 ,mark options={solid, fill=red!100}, red!100]
    table[y=minGopt,x=L]{Figs/DoF_Prop_SOTA.tex};  
    \addlegendentry{\normalsize \textit{min-$G$}};

    \addplot
    [mark = diamond,line width=.7pt, mark size=4.4 , cyan!100]
    table[y=groupingop,x=L]{Figs/DoF_Prop_SOTA.tex};
    \addlegendentry{\normalsize \textit{Grouping} };
    

    \addplot
    [dashed,line width=.7pt, mark = star, mark size=4.8 , blue!100]
    table[y=groupingopeff,x=L]{Figs/DoF_Prop_SOTA.tex};
    \addlegendentry{\normalsize \textit{Super-grouping} };


    \addplot
    [dash dot,line width=.7pt, mark = triangle , mark size=4.4,mark options={solid, fill=black!100} , black!100]
    table[y=phantomopt,x=L]{Figs/DoF_Prop_SOTA.tex};
    \addlegendentry{\normalsize \textit{Phantom} };
    

    \end{axis}

    \end{tikzpicture}
    }
    \vspace{-2mm}
    \caption{DoF scaling versus Tx multiplexing gain, $\mathrm{opt}$ policy, $\gamma$ = 0.05,
    $J=5$, $\{K_j\}_{j=1}^J=\{20,20,20,40,80\}$, and $\{G_j\}_{j=1}^J=\{2,4,6,8,12\}$.
    }
    \label{fig:Tx_side_scaling}
\end{figure}

In Fig.~\ref{fig:Plot_scaling_gamma}, the achievable DoF is shown as a function of the cache size
$\gamma$ under the $\mathrm{opt}$ policy, for two Rx-antenna profiles (blue and red colors, as specified in the caption).
As $\gamma$ increases, all strategies benefit from enhanced coded-multicasting
opportunities enabled by larger caches. \emph{Phantom} and \emph{Super-grouping} strategies consistently extract larger gains from increasing $\gamma$ than
\emph{Grouping} and \emph{min-$G$}, as they more effectively translate caching into parallel streams under heterogeneous receiver-side spatial multiplexing gain configurations. 
%
%
For both profiles, \emph{Phantom} and
\emph{Super-grouping} outperform the primary
strategies, whose relative ordering depends on the cache size and Rx-antenna configuration.
%
%

\begin{figure}[t]
    \centering

    \resizebox{.8\columnwidth}{!}{%

    \begin{tikzpicture}

    \begin{axis}
    [
    axis lines = center,
    xlabel near ticks,
    xlabel = \smaller {$\gamma$},
    ylabel = \smaller {Achievable DoF},
    ylabel near ticks,
    ymin = 28,
    ymax = 116,
    xmin = 0.018,  
    xmax = 0.1,
     xticklabel style={/pgf/number format/fixed}, 
    legend style={
    nodes={scale=0.6, transform shape},
    at={(.33,1.0)}, 
    anchor=north, 
    draw=black,
    outer sep=2pt,
    font=\normalsize,
    legend columns=2, 
    scale=0.7, 
    text height=3.0ex, text depth=.25ex 
    },
    ticklabel style={font=\smaller},
    grid=both,
    minor tick num = 4, 
    major grid style = {line width=0.5pt,draw=gray!50},
    minor grid style = {line width=0.2pt,draw=gray!30}
    ]

    \addplot
     [dash dot,line width=.7pt, mark = x, mark size=3.6,mark options={solid, fill=red!100} , red!100]
    table[y=minGopt,x=gamma]{Figs/DoF_Prop_gamma5.tex};
    \addlegendentry{\normalsize \textit{min-$G$}
    };

    
    \addplot
     [line width=.7pt,mark = star , mark size=3.6, red!100]
    table[y=groupingop,x=gamma]{Figs/DoF_Prop_gamma5.tex};
    \addlegendentry{\normalsize \textit{Grouping} };


    \addplot
    [dashed,line width=.7pt, mark = o, mark size=3.6,mark options={solid, fill=red!100}, red!100]
    table[y=groupingopeff,x=gamma]{Figs/DoF_Prop_gamma5.tex};
    \addlegendentry{\normalsize \textit{Super-grouping} };

    
    \addplot
    [dash dot,line width=.7pt, mark = square, mark size=3.6,mark options={solid, fill=red!100} , red!100]
    table[y=phantomopt,x=gamma]{Figs/DoF_Prop_gamma5.tex};
    \addlegendentry{\normalsize \textit{Phantom} };
    
\addplot
     [dash dot,line width=.7pt, mark = +, mark size=3.6,mark options={solid, fill=blue!100} , blue!100]
    table[y=minGopt,x=gamma]{Figs/DoF_Prop_gamma6.tex};
    \addlegendentry{\normalsize \textit{min-$G$}
    };

    
    \addplot
     [line width=.7pt,mark = star , mark size=3.6, blue!100]
    table[y=groupingop,x=gamma]{Figs/DoF_Prop_gamma6.tex};
    \addlegendentry{\normalsize \textit{Grouping} };


    \addplot
    [dashed,line width=.7pt, mark = o, mark size=3.6,mark options={solid, fill=blue!100} , blue!100]
    table[y=groupingopeff,x=gamma]{Figs/DoF_Prop_gamma6.tex};
    \addlegendentry{\normalsize \textit{Super-grouping} };

    
    \addplot
    [dash dot,line width=.7pt, mark = square, mark size=3.6,mark options={solid, fill=blue!100} , blue!100]
    table[y=phantomopt,x=gamma]{Figs/DoF_Prop_gamma6.tex};
    \addlegendentry{\normalsize \textit{Phantom} };

    \end{axis}

    \end{tikzpicture}
    }
    \vspace{-2mm}
    \caption{DoF scaling versus $\gamma$, $\mathrm{opt}$ policy; $J=5$, $L=16$, and $\{K_j\}_{j=1}^J=\{50, 50$,$50, 120, 100\}$. For blue curves: $\{G_j\}_{j=1}^J=\{2, 4,8, 10,14\}$, and for red curves: $\{G_j\}_{j=1}^J = \{2, 6,8, 10,14\}$. 
    }
    \label{fig:Plot_scaling_gamma}
\end{figure}
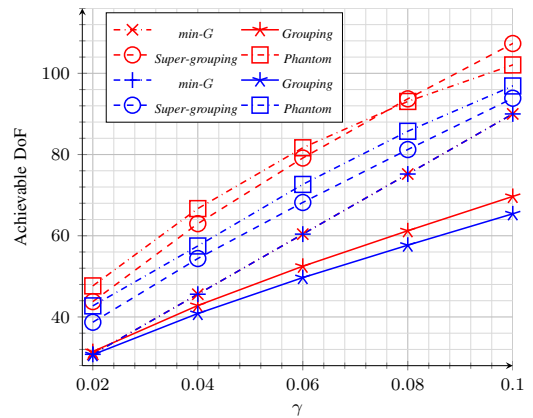
In Fig.~\ref{fig:plot_scal_Rx_SM}, the scalability of the achievable DoF is examined for 
different strategies under varying Rx-side spatial multiplexing gains (simulation setups $\mathrm{sc}_1$--$\mathrm{sc}_7$ defined in Table~\ref{tab:sim_setups}). 
The results indicate that the achievable DoF scaling is not determined by $\{G_j\}$ alone, but also by how the antenna profile aligns with the
user distribution across groups, as observed in $\mathrm{sc}_1$--$\mathrm{sc}_4$. In contrast, for $\mathrm{sc}_5$--$\mathrm{sc}_7$, where $\{K_j\}$ is fixed and the Rx antenna profile becomes overall larger in terms of antenna counts, both the
\emph{Super-grouping} and \emph{Phantom} strategies achieve a larger DoF by effectively exploiting the increased Rx-side spatial multiplexing gains.
\begin{figure}[t]
    \centering
    \resizebox{.8\columnwidth}{!}{%

   
    \begin{tikzpicture}
\begin{axis}[
    axis lines = left,
    xlabel = \smaller {Scenario},
    ylabel = \smaller {Achievable DoF},
    ylabel near ticks,
    xlabel near ticks,
    legend pos = north west,
    legend columns=3,
    legend style = {
        nodes={scale=0.99, transform shape},
        cells={align=center},
        at={(0.02,1)},
        anchor=north west,
        font=\small
    },
    ybar,
    legend image code/.code={%
        \draw[#1] (0cm,-0.045cm) rectangle (0.42cm,0.25cm);
        },
    ymin = 0,
    ymax =77,
    xtick=data,
    enlarge x limits=0.1, 
    symbolic x coords={sp1,sp2,sp3,sp4,sp5,sp6,sp7},
    xticklabels={$\mathrm{sc}_1$,$\mathrm{sc}_2$,$\mathrm{sc}_3$,$\mathrm{sc}_4$,$\mathrm{sc}_5$,$\mathrm{sc}_6$,$\mathrm{sc}_7$},
    ticklabel style={font=\smaller},
    grid=both,
    major grid style={line width=.3pt,draw=gray!30},
    bar width=0.2cm,
]
 \addplot 
    [black,fill=brown!30]
    table
    [x=Scen ,y=minGopt]
    {Figs/DoF_Scalability3.tex};
    \addlegendentry{\smaller \textit{min-$G$}}

    \addplot 
    [black,fill=black!50]
    table
    [x=Scen ,y=groupingopeff]
    {Figs/DoF_Scalability3.tex};
    \addlegendentry{\smaller \textit{Super-grouping}}
    
    \addplot 
    [black,fill=red!50]
    table
    [x=Scen ,y=phantomopt]
    {Figs/DoF_Scalability3.tex};
    \addlegendentry{\smaller \textit{Phantom}}
    


    \end{axis}
    \end{tikzpicture}
    }
    \vspace{-2mm}
    \caption{DoF scaling versus Rx antenna distribution, $\mathrm{opt}$ policy. 
    }
    \label{fig:plot_scal_Rx_SM}
\end{figure}
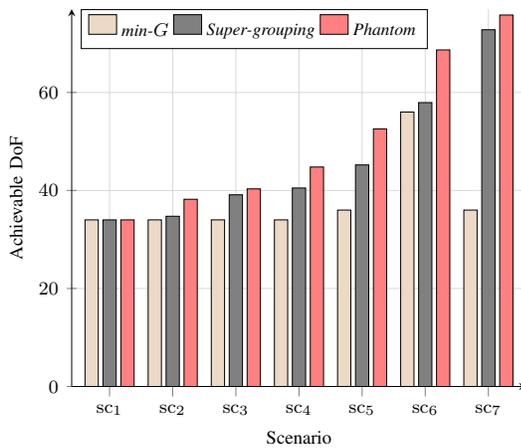
\section{Conclusion}
\label{sec:conclusions}
This paper characterized the achievable single-shot DoF \color{black}{lower bounds} \color{black} of cache-aided MIMO systems with \emph{heterogeneous} per-user receive antennas. Building on three symmetric reference policies $\pi\!\!\!\in \!\!\!\{\mathrm{opt},\mathrm{cmb},\mathrm{lin}\}$, we developed four asymmetric delivery strategies that yield distinct operating points between global coded-multicasting and per-user spatial multiplexing, with different achievable DoF and subpacketization requirements.
The two ``primary'' strategies \emph{min-$G$} and \emph{Grouping} expose the fundamental tradeoff between the global caching and spatial multiplexing gains, while the two ``hybrid'' strategies, \emph{Super-grouping} and \emph{Phantom}, further expand the achievable region by combining these primary principles in a structured way. 
We derived closed-form \color{black}and finite-search-optimized achievable \color{black} DoF expressions and established linear decodability for all policy--strategy combinations. 

\appendix
%
%
%
%
%
%
\subsection{DoF Gap between Ceiling-based Point and \emph{$\rm opt$} Policy}
\label{app_A0}
\begin{lem}\label{lm:ceiling_opt_cmp}
Assume that the ceiling-based point in Remark~\ref{opt_special_case_ceiling} is feasible. Then, $\rm DoF_{\rm ceil}\!\leq\! \rm DoF_{\rm opt}$, and the achievable DoF gap $\Delta_{\rm gap }$ between the two policies 
satisfies $0\!\leq\! \Delta_{\rm gap }\!\leq \!(G-3)^+$.     
\end{lem}
\begin{proof}
     Let $n\triangleq\big\lceil \tfrac{L}{G} \big\rceil \geq 1$. 
    Then, \(L\) can be written as $L \triangleq (n-1)G + r$, where $1 \!\leq r \!\leq\! G$\footnote{\color{black}Indeed, \(n-1\!<\!\frac{L}{G}\leq \!n\), and hence \(0\!<\!L-(n-1)G\!\leq\! G\). Since \(L\) and \(G\) are integers, \(1\!\leq \!r\leq\! G\). In particular, \(r=G\) when \(L\) is divisible by \(G\).}. 
    From~\cite{naseritehrani2026cache}, the
\(\mathrm{opt}\) policy achieves
    \begin{align}
\mathrm{DoF}_{\mathrm{opt}} &= \max_{\Omega,\beta}\Omega\beta \label{eq:total_DoF12} \\[-.6ex]
\color{black}\mathrm{s.t.}\color{black} \quad 
\text{i)} \: &\beta \le G, \: \: \text{ii)} \: \Big\lceil  \nicefrac{ \beta}{\binom{\Omega-1}{t}} \Big\rceil \le L - (\Omega-t-1) \beta \nonumber
\end{align}
For the gap analysis, we distinguish the cases \(\beta^*=G\) and
\(\beta^*<G\).
If $\beta^*=G$,
condition~ii) implies that the maximum achievable value of $\mathrm{opt}$ policy is $\Omega^*\leq t+\big\lceil \nicefrac{L}{G} \big\rceil$. Since the endpoint is feasible by Remark~\ref{opt_special_case_ceiling}, the resulting DoF is thus $
\Omega^*\beta^* = G(t+\lceil \nicefrac{L}{G}\rceil)$. Hence, \(\Delta_{\rm gap}=0\). It remains to consider $\beta^*< G$. 
Now, we can  
upper-bound the DoF gap between $\mathrm{DoF}_{\mathrm{opt}}$ and \(\mathrm{DoF}_{\!\rm ceil} = G\big(t+\lceil \nicefrac{L}{G}\rceil\big)= G(t+n) \) as: 
\begin{equation}\label{eq:total_DoF2}
 \begin{array}{l}
  0\leq {\mathrm{DoF}_{\rm opt}}-{\mathrm{DoF}_{\!\rm ceil}} 
  \overset{(a)}{\leq} 
  L-1+\beta^* (t+1)-G(t+n),\\[.1ex] 
   \hspace{33.5mm}\overset{(b)}{=} r-1+(\beta^*-G)(t+1) \\[.1ex]
   \hspace{33.5mm}\overset{(c)}{\leq} r-t- 2 \\[.1ex]
   \hspace{33.5mm}\overset{(d)}{\leq} G-3\\[-2mm]
 \end{array}
\end{equation}
Here, \((a)\) follows from condition~ii). Specifically, since $\Big\lceil \nicefrac{\beta^*}{\binom{\Omega^*-1}{t}} \Big\rceil \geq 1$, R.H.S of condition~ii) must also satisfy $L-(\Omega^*-t-1)\beta^*\geq1$. Rearranging this inequality yields $\Omega^*\beta^*\leq L-1+\beta^*(t+1)$.
\((b)\) follows from \(L \triangleq (n-1)G + r \), and \((c)\) follows from \(\beta^* \leq G-1\), since \(\beta^*<G\) and \(\beta^*\in\mathbb Z_+\), and \((d)\) follows from $r\leq G$ and $t \geq 1$. 


Hence, when $\beta^*<G$, \eqref{eq:total_DoF2} yields $\Delta_{\rm gap}\leq G-3$. This case can lead to a nonnegative gap only when $G\geq3$; for $G=3$, nonnegativity of $\Delta_{\rm gap}$ therefore forces $\Delta_{\rm gap}=0$. For \(G\leq2\), the case \(\beta^*<G\) leads to a contradiction, and hence
\(\beta^*=G\), which gives \(\Delta_{\rm gap}=0\). 
Consequently, in all cases,
$0\leq \Delta_{\rm gap}\leq (G-3)^+$,
which completes the proof.
\end{proof}
\subsection{Proof of Lemma~\ref{LD_Phantom}}
\label{app_B}
\color{black}
    In each MC round \(i\in[I]\), Step~2 first constructs the candidate multicast
transmission vectors \(\tilde{\Bx}_{\pi,i}^{\mathrm{MC}}(s)\) in the
hypothetical symmetric MIMO-CC setup, where every user has \(\hat G_i\) receive
dimensions and receives \(\hat\beta_{\pi,i}\) streams. The selected value
\(\hat\beta_{\pi,i}\) satisfies the corresponding linear decodability
condition~\eqref{DoF_bndd}\footnote{\color{black}Calculations are provided here for the $\hat \Omega_{\pi, i} > K\gamma+1$ case only. Extension to the equality case is straightforward.
}:
\begin{equation}\label{ld_ref}
\hat{\beta}_{\pi,i}
\!\leq \!
\min \!
\big(
\hat G_i,\,
\nicefrac{\!\!\!\displaystyle(L-\delta_{\pi,i})}{\Gamma_{\pi,i}}
\big)
\!\!\Longleftrightarrow \!\!
\begin{cases}
\!\text{\rm(Ia)}\:
\delta_{\pi,i}
\leq
L-\Gamma_{\pi,i}\hat{\beta}_{\pi,i},\\
\!\text{\rm(IIa)}\quad
\hat{\beta}_{\pi,i}\leq\hat G_i.
\end{cases}
\end{equation}
Here, \(\delta_{\pi,i}\) denotes the maximum number of subpackets \(W_{\wp_\pi}^{q}\) having the same packet index \(\wp_\pi\) in each transmission interval of MC round \(i\), and \(\Gamma_{\pi,i} \triangleq \hat \Omega_{\pi, i}- K\gamma-1\) denotes 
the number of unintended users at which each considered subpacket must
be nulled. 
 Condition~\text{\rm(Ia)} is the transmit-side decodability condition:
for the \(\delta_{\pi,i}\) subpackets having the same packet index,
\(\delta_{\pi,i}\) linearly independent beamformers must be selected from the
nullspace remaining after nulling
\(\Gamma_{\pi,i}\hat{\beta}_{\pi,i}\) interfering streams.
Condition~\text{\rm(IIa)} is the receive-side condition in the hypothetical
symmetric setup.

We next move from the hypothetical symmetric setup to the actual asymmetric
setup. For a subpacket \(W_{\wp_\pi}^{q}\) intended for user \(k\), its
beamformer \(\Bw_{\wp_\pi}^{q}\) must null the inter-stream interference at
every active stream of each user \(k'\in\CI_{\wp_\pi,k}(s)\), whereas the
remaining interference is removed using the cache contents.
%
%
Let $\BU_{k'}\in \mathbb C^{\SfG_{k'}\times \bm\beta_{\pi,i}^{k'}(s)}$
be the receive combining matrix of user \(k'\).
Then,
$\BH_{k'}^{H}\BU_{k'}
\in
\mathbb C^{L\times \bm\beta_{\pi,i}^{k'}(s)}$.
For transmission $s$, the equivalent interference channel associated with \(W_{\wp_\pi}^{q}\) and included in $\hat{\Bx}_{\pi,i}^{\mathrm{MC}}(s)$ is obtained by
concatenating 
the effective interference channels of the users in
\(\CI_{\wp_\pi,k}(s)\), namely\footnote{\color{black}For $\pi\in\{\mathrm{opt},\mathrm{cmb}\}$, 
$\CI_{\wp_\pi,k}(s)
=\hat{\CK}_{\pi,i}(s)\setminus
\bigl(\CP\cup\{k\}\bigr)$, while, for $\pi=\mathrm{lin}$, 
$\CI_{p,k}(s)
=
\left\{
k'\in\hat{\CK}_{\mathrm{lin},i}(s)\setminus\{k\}:
V[p,k']=0
\right\}$,
where $V[p,k]$ denotes the $(p,k)$-th element of the cyclic placement
matrix $\mathbf V$ in~\cite{salehi2020lowcomplexity}.\color{black}} 
\begin{equation}\label{ld_0}
\bar{\BH}_{\wp_\pi,k}(s)
=
\left[
\BH_{k'}^{H}\BU_{k'}
\right]^H_{k'\in\CI_{\wp_\pi,k}(s)}
\in
\mathbb C^{r_{\wp_\pi,k}(s)\times L},
\end{equation}
where
$r_{\wp_\pi,k}(s)
=
\sum_{k'\in\CI_{\wp_\pi,k}(s)}
\beta_{\pi,i}^{k'}(s)$
is the number of streams zero-forced by the beamformer of
\(W_{\wp_\pi}^{q}\).
In the actual transmission, user \(k\) receives
$\bm\beta_{\pi,i}^{k}(s)
=
\min(\hat{\beta}_{\pi,i},\SfG_k)$
streams. Hence,
(Ib) $\bm\beta_{\pi,i}^{k}(s)\leq \hat{\beta}_{\pi,i}$,
(IIb) $\bm \beta_{\pi,i}^{k}(s)\leq \SfG_k$.
Condition~\text{\rm(Ib)} shows that removing unsupported streams in Step~3 cannot increase the number of active streams, whereas
condition~\text{\rm(IIb)} guarantees that user \(k\) is not assigned more streams than its receive dimension. 
%
Consequently,
\begin{equation}\label{ld_2}
r_{\wp_\pi,k}(s)
\leq
\Gamma_{\pi,i}\hat{\beta}_{\pi,i}.
\end{equation}

For each transmitted subpacket \(W_{\wp_\pi}^{q}\), the corresponding beamformer must satisfy
\begin{equation}\label{ld_3}
\Bw_{\wp_\pi}^{q}
\in
\mathrm{Null}
\big(
\bar{\BH}_{\wp_\pi,k}(s)
\big).
\end{equation}
Hence, by the rank-nullity theorem,
\begin{equation}\label{ld_3_nullity}
\rm nullity
\big(
\bar{\BH}_{\wp_\pi,k}(s)
\big) \color{black}= 
L-\rm rank \big(
\bar{\BH}_{\wp_\pi,k}(s)
\big)
\ge \color{black}
L-r_{\wp_\pi,k}(s).
\end{equation}
Let \(\delta_{\pi,i}^{k}(s)\) denote the number of remaining subpackets
intended for user \(k\) that have the same packet index \(\wp_\pi\). Their
beamformers must be selected linearly independent from the nullspace in
\eqref{ld_3}.
Therefore, a sufficient linear decodability condition is
\begin{equation}\label{ld_4}
\delta_{\pi,i}^{k}(s)
\leq
\mathrm{nullity}
\big(
\bar{\BH}_{\wp_\pi,k}(s)
\big).
\end{equation}
Moreover, removing unsupported streams cannot increase the number of subpackets having the same packet index. Thus,
\begin{equation}\label{ld_5}
\delta_{\pi,i}^{k}(s)\leq \delta_{\pi,i}.
\end{equation}
Combining~\eqref{ld_ref}, \eqref{ld_2}, and \eqref{ld_5}, we obtain
\begin{equation}\label{ld_6}
\delta_{\pi,i}^{k}(s)
+
r_{\wp_\pi,k}(s)
\leq
\delta_{\pi,i}
+
\Gamma_{\pi,i}\hat{\beta}_{\pi,i}
\leq
L.
\end{equation}
Equivalently,
\begin{equation}\label{ld_7}
\delta_{\pi,i}^{k}(s)
\leq
L-r_{\wp_\pi,k}(s)
\color{black}\leq\color{black}
\mathrm{nullity}
\big(
\bar{\BH}_{\wp_\pi,k}(s)
\big).
\end{equation}
Hence, the remaining subpackets having the same packet index fit within the available nullspace, and their beamformers can be selected linearly independent while satisfying all required ZF constraints.
Finally, $\bm\beta_{\pi,i}^{k}(s)
=
\min(
\hat{\beta}_{\pi,i},
\SfG_k)
\leq
\SfG_k$,
and therefore every user has sufficient receive dimensions for its assigned streams. Consequently, every actual Phantom multicast transmission vector remains linearly decodable for all users in the asymmetric setup.
\color{black}  

\bibliographystyle{IEEEtran}
\bibliography{conf_short,IEEEabrv,references,whitepaper}
 
\end{document}